\documentclass[journal,twoside,web]{ieeecolor}
\usepackage{generic}
\usepackage{textcomp}
\usepackage{algorithmic}

\let\labelindent\relax

\usepackage{enumitem}

\usepackage{amsmath,amssymb,amsfonts}
\usepackage{cite}
\usepackage[colorlinks = true, linkcolor = blue, citecolor =
blue]{hyperref}

\usepackage[english]{babel}
\usepackage{subcaption}
\usepackage{graphicx}

\newcommand{\Mc}[1]{\mathcal{#1}}
\newcommand{\agt}{\mathcal{V}} 
\newcommand{\agtiso}{\mathcal{V}_{cl}} 
\newcommand{\agtmmax}{\mathcal{V}_{\max}}
\newcommand{\agtmmin}{\mathcal{V}_{\min}}

\newcommand{\zi}{z_{i}} 
\newcommand{\zk}{z_{k}}
\newcommand{\z}{z}
\newcommand{\znoti}{\mathbf{z}_{-i}}

\newcommand{\zbari}{\bar{z}_{i}}

\newcommand{\aik}{a_{ik}}

\newcommand{\aij}{a_{ij}}

\newcommand{\boldz}{\mathbf{z}}
\newcommand{\boldp}{\mathbf{p}}

\newcommand{\ri}{r_{i}} 

\newcommand{\rj}{r_{j}}
\newcommand{\di}{p_{i}} 

\newcommand{\wi}{w_{i}} 

\newcommand{\Sf}{S}
\newcommand{\Cf}{C}
\newcommand{\adjmat}{\mathbf{A}}
\newcommand{\Ni}{\mathcal{N}_{i}}
\newcommand{\Nie}{\mathcal{N}_{i}^{e}}
\newcommand{\Nif}{\mathcal{N}_{i}^{f}}
\newcommand{\agtsumk}{\sum_{k \in \agt}}

\newcommand{\enemysumi}{\sum_{k \in\mathcal{N}_{i}^{e}}}
\newcommand{\friendsumi}{\sum_{k \in \mathcal{N}_{i}^{f}}}

\newcommand{\agtsumknoti}{\sum_{k \in \agt \setminus \left\{i\right\}}}
\newcommand{\agtsumknotj}{\sum_{k \in \agt \setminus \left\{j\right\}}}
\newcommand{\poa}{\mathsf{PoA}}
\newcommand{\poae}{\pi_{e}}
\newcommand{\poau}{\pi_{u}}
\newcommand{\costi}{\chi_{i}}
\newcommand{\costj}{\chi_{j}}
\newcommand{\satratioi}{\mathsf{SR}_{i}}
\newcommand{\realn}{\mathbb{R}^{n}} 
\newcommand{\real}{\mathbb{R}} 
\newcommand{\nat}{\mathbb{N}} 
\newcommand{\boldone}{\mathbf{1}}
\newcommand{\boldzero}{\mathbf{0}}
\newcommand{\trpose}{^{\top}}

\newcommand{\Ce}{\mathsf{C}_{E}}
\newcommand{\Cu}{\mathsf{C}_{U}}

\newcommand{\edg}{\mathcal{L}}

\newcommand{\ldef}{:=}
\newcommand{\rdef}{=:}

\def \mmax{m_{\max}}
\def \mmin{m_{\min}}

\newcommand{\thmtitle}[1]{\mbox{}\textit{(\textbf{#1}.)}}

\newcommand{\remend}{\relax\ifmmode\else\unskip\hfill\fi\hbox{$\bullet$}}

\newcommand{\bulletsym}{\hbox{$\bullet$}}
\newcommand{\bulletend}{\relax\ifmmode\else\unskip\hfill\fi\bulletsym}
\newcommand{\squaresym}{\hbox{$\blacksquare$}}
\newcommand{\proofend}{\relax\ifmmode\else\unskip\hfill\fi\squaresym}
\newcommand{\trianglesym}{\hbox{$\blacktriangle$}}
\newcommand{\egend}{\relax\ifmmode\else\unskip\hfill\fi\trianglesym}

\usepackage{tcolorbox}

\def \eqpt{\mathcal{E}}				
\def \eqptp{\mathcal{E}(\mathbf{p})}

\def \grph{\mathsf{G}}
\def \lnep{\mathcal{NE}_{l}(\mathbf{p})}

\def \lne{\mathcal{NE}_{l}}
\def \nep{\mathcal{NE}(\mathbf{p})}

\def \ne{\mathcal{NE}}
\def \one{\mathbf{1}}

\newcommand{\tth}{^{\text{th}}}

\renewenvironment{proof}{\textit{Proof:} }{\proofend}

\newtheorem{theorem}{\textbf{Theorem}}[section]
\newtheorem{corollary}[theorem]{\textbf{Corollary}}
\newtheorem{lemma}[theorem]{\textbf{Lemma}}
\newtheorem{proposition}[theorem]{\textbf{Proposition}}
\newtheorem{example}[theorem]{\textbf{Example}}

\newtheorem{remark}[theorem]{\textbf{Remark}}
\newtheorem{define}[theorem]{\textbf{Definition}}

\renewcommand{\baselinestretch}{0.973}

\def\BibTeX{{\rm B\kern-.05em{\sc i\kern-.025em b}\kern-.08em
    T\kern-.1667em\lower.7ex\hbox{E}\kern-.125emX}}
\begin{document}
\title{\LARGE \bf Opinion Dynamics for Utility Maximizing Agents: Exploring the Impact of the Resource Penalty}
\author{Prashil Wankhede$^{1}$, Nirabhra Mandal$^{2}$, Sonia
  Mart\'{i}nez$^{2}$ and Pavankumar Tallapragada$^{1}$
  \thanks{This work was partially supported by Science and Engineering
    Research Board under grant CRG/2023/008573. $^{1}$Prashil Wankhede
    and Pavankumar Tallapragada are with the Department of Electrical
    Engineering, Indian Institute of Science, Bangalore. \{\tt\small \{prashilw,
    pavant\}@iisc.ac.in\}} \thanks{$^{2}$Nirabhra Mandal and Sonia
    Mart\'{i}nez are with the Department of Mechanical and Aerospace
    Engineering, University of California, San Diego.  \{\tt\small
    \{nmandal, soniamd\}@ucsd.edu\}} }
\maketitle
\begin{abstract}
  We propose a continuous-time nonlinear model of opinion dynamics
  with utility-maximizing agents connected via a social influence
  network. A distinguishing feature of the proposed model is the
  inclusion of an opinion-dependent resource-penalty term in the
  utilities, which limits the agents from holding opinions of large
  magnitude.   
  This model is applicable in scenarios where the opinions pertain
  to the usage of resources, such as money, time, computational resources etc.
  Each agent myopically seeks to maximize its utility by revising its
  opinion in the gradient ascent direction of its utility function,
  thus leading to the proposed opinion dynamics.  We show that, for
  any arbitrary social influence network, opinions are ultimately
  bounded. For networks with \emph{weak antagonistic relations}, we
  show that there exists a globally exponentially stable equilibrium
  using contraction theory.  We establish conditions for the existence
  of consensus equilibrium and analyze the relative dominance of the
  agents at consensus. We also conduct a game-theoretic analysis of
  the underlying opinion formation game, including on Nash equilibria
  and on prices of anarchy in terms of \emph{satisfaction
    ratios}. Additionally, we also investigate the oscillatory
  behavior of opinions in a two-agent scenario. Finally, simulations
  illustrate our findings.
\end{abstract}
	
\begin{IEEEkeywords}
  Opinion dynamics, Multi-agent systems, Utility maximization, Game
  theory.
\end{IEEEkeywords}
	
\section{Introduction}
\label{sec:into_and_notation}
\IEEEPARstart{O}{pinion} dynamics deals with the modeling and
mathematical analysis of how beliefs and ideas evolve and spread
within social groups or networks over time. As collective opinions
have far-reaching implications in diverse sectors such as policy,
public health, sociology, finance or economics, it is paramount to
understand their drivers and consequences. Even for smaller groups,
comprehending opinion forming is a necessary first step toward the
management of mixed multi-agent groups, their decision making and
subsequent dynamic interactions.  While many simpler opinion dynamic
models have been developed, modeling and analysis of networked
rational agents who react to neighbors' influence remains limited. 
In this work, we pay particular attention to the setting where the
opinions of the agents are related to the usage of resources for a
particular task. Each agent has heterogeneous resources available to
them which limits their opinions about resource usage. We capture
this in our model by including a resource penalty term in the
utility functions.

\subsubsection*{Literature Review}

Many developments in the opinion dynamics literature stem from classical models, including the French-DeGroot (FD) model~\cite{1956_French, 1974_DeGroot}, the Abelson model~\cite{1964_Abelson}, the Taylor model~\cite{1968_Taylor}, the Friedkin-Johnsen (FJ) model~\cite{1990_NF-EJ}, the Hegselmann-Krause (HK) model~\cite{2002_RH-UK}, the Altafini model~\cite{2012_Altafini}, and the Deffuant-Weisbuch (DW) model~\cite{2001_GD-etal}.

For an in-depth summary of the above
models and survey of the literature, we refer the readers to
\cite{2017_AP-RT_ARC,2018_AP-RT_ARC,2020_FD_BA_MY_AP_CR,2020_QZ_GK_HZ_HL_XC_CCL_YD}. These
fundamental models serve as the foundation for many others in the
literature. Some examples of recent models include the discrete-time
versions of the Altafini model with time-varying signed
networks~\cite{2016_ZM_GS_KH_MC_YH,2017_JL_XC_TB_MB}. The
work~\cite{2021_PCV_KC_FB} proposes the \emph{affine boomerang}
opinion dynamics model with asynchronous opinion updating,
incorporating the phenomenon of the \emph{boomerang} effect into the
dynamics. Unlike~\cite{2012_Altafini,2016_ZM_GS_KH_MC_YH,2017_JL_XC_TB_MB},
in this model~\cite{2021_PCV_KC_FB}, the opinions do not converge to
zero if the network does not satisfy a \emph{structural balance}
property, but rather exhibits bounded
fluctuations. Reference~\cite{2019_MY_YQ_BDOA__MC_AG} proposes the
\emph{expressed-private-opinion} model in which every agent's private
opinion is influenced by the expressed opinions of its neighbors.  The
concept of \emph{social power} was first introduced
in~\cite{1956_French} in order to identify the most influential agent
in the social network. The DeGroot-Friedkin (DF)
model~\cite{2015_PJ-etal} integrates the opinion formation process
with social power evolution through a reflected self-appraisal
mechanism. A generalized DF model was proposed in~\cite{2019_MY_BA}
and an extension of the DF model to stubborn agents was proposed
in~\cite{2021_AM_FB_PJ_YT_NF_LW}.
Some works on opinion dynamics~\cite{2021_PCV_KC_FB,2023_AB_AF_NL}
also analyze fluctuations and periodic behavior of opinions.

The modeling of opinion dynamics through a game-theoretic or a
utility maximization approach is still in its early stages,
with only a limited number of works published thus far. 
Reference~\cite{2009_JM_GA_JS} showed that cooperative control
problems, such as consensus seeking, can be effectively tackled
using game-theoretic methods, particularly through potential and
weakly acyclic games. The works \cite{2014_JG_SR,2015_Bindel_et_al}
interpret the FJ-model as best response dynamics within a
game-theoretic framework, where each agent aims to minimize a cost
function. The work~\cite{2014_PG_JL_FS} introduced a general
framework for social influence grounded in the psychological concept
of cognitive dissonance and demonstrated that various opinion
dynamics models can be viewed as best response dynamics in a
coordination game, where utilities are determined by dissonance
functions.
In a recent study \cite{2021_SRE}, a
dynamic influence maximization game is explored, where a set
of competing \emph{players} allocate their fixed resources
over certain \emph{individuals} (who hold opinions about
\emph{players}) to maximize their utilities in the long
term. Reference \cite{2015_SRE-TB} performs a game-theoretical
analysis of the asynchronous HK model. The
works~\cite{2020_LZ_MY_MC,2023_HDA_MY_LZ_ZC_MC} capture the co-evolution of
opinions and actions taken by the agents under a game-theoretic framework. The
work~\cite{2021_SP-etal} utilizes a continuous-time non-linear
opinion dynamics model to tune the mutual cooperative behavior
of agents in a repeated game. Within this framework, agents
make strategic decisions relying upon rationality and
reciprocity. Another approach, detailed in
\cite{2018_RE_TB_AN_VP_SB}, introduces a discrete-time opinion
dynamics model with a game-theoretic structure where the agents
incur a combination of conformity and manipulation costs based
on the opinions. The aim of each agent is to minimize this
cost.

\par Exploration of how agents' resources impact their opinions and
social influencing capabilities is currently an open question in the
opinion dynamics literature. 
Motivated by this, we adopt a utility maximization and game-theoretic framework in the current work to investigate the limiting effects of agents' resources on their opinions about its usage. A preliminary version of this work appeared in
\cite{2023_PW_NM_PT}, where we assumed the underlying social influence
network to be complete. The current work extends \cite{2023_PW_NM_PT}
to the case of any general social network topologies and also allows
for pairwise antagonistic relationships among the agents. We also
investigate several new questions, such as exponential stability of
the equilibria, price of anarchy and periodic evolution of opinions.
	
\subsubsection*{Contributions}
The main contributions of this work are:

\begin{enumerate}[wide=\parindent]
  \item We define agent's utility function to capture the tradeoffs of
    internal opinion preferences, attachment (or stubbornness) toward
    its own opinion, conformity or non-conformity and lastly a
    resource penalty term. The resource penalty term in the utility
    function depends on the agent's resources, which is a
    representation of the agent's wealth, time etc. In this work, we
    assume that agents are involved in the opinion formation process
    to make some decisions about the use of their resources for a
    particular task. In the
     literature, most models of opinion dynamics are generic and seem
     to assume that the topic of discussion itself has no effect on
     the opinions' evolution. 
    However, this is not realistic in our
    setting.
    Thus, we include a resource penalty term in the utility that keeps
    the opinions bounded.
    We propose a novel opinion dynamics model from the utility
    functions based upon the assumption that every agent myopically
    seeks to maximize its own utility. We refer to the underlying game
    as the \emph{opinion formation game}.

\item Unlike the existing works which consider stubborn agents, under
  the proposed dynamics agents can reach consensus even if their
  internal preferences are different. If the agents' opinions reach a
  consensus equilibrium, we can use the \emph{consensus dominance
    weights} of the agents to deduce the relative \emph{dominance} of
  each agent. The consensus dominance weights depend on the resources
  of the agents.
		
\item We conduct a game-theoretic analysis of the opinion formation
  game.  Our Nash equilibrium conditions hold even when the network
  consists of \emph{antagonistic} relationships where opinions of some
  agents negatively influence the opinions of others.  We provide a
  relation between the set of local Nash equilibria $\lne$, set of
  Nash equilibria $\ne$ of the opinion formation game and the set of
  equilibrium points $\eqpt$ of the opinion dynamics. Specifically, we
  show that $\ne \subseteq \lne \subseteq \eqpt$. Further, we give a
  condition for these sets to coincide.
		
\item In the case of \emph{weak antagonistic} relationships, we show
  that the game exhibits a unique Nash equilibrium and the agents'
  opinions converge to it under the proposed dynamics starting from
  any arbitrary initial opinion profile. 
  A special case of \emph{weak
    antagonistic relationships} is one where antagonistic
  relationships are absent in the social network.

\item In the case where antagonistic relationships are absent, we
  bound the Price of Anarchy in terms of the \emph{satisfaction
    ratios} to quantify the inefficiency of the unique Nash
  equilibrium.  The satisfaction ratio of an agent at a given opinion
  profile of all agents is the ratio of the utility received by the
  agent at this opinion profile to the maximum possible utility; thus
  quantifying how ``satisfied'' that agent is with that opinion
  profile. Using these bounds, we show that if agents opinions
  converge to a consensus then it is a socially optimal outcome.
		
\item We analyze oscillatory and periodic opinion behavior.  We show
  that it is necessary for the agents to have sufficiently strong antagonistic relations for the
  two-agent dynamics to have periodic solutions. We also provide
  sufficient conditions for a Hopf bifurcation to exist for the
  two-agent dynamics.
\end{enumerate}
	
\subsubsection*{Notation}

Throughout the paper, we use non-bold letters for denoting scalars,
bold lowercase letters for denoting vectors, and bold uppercase
letters for denoting matrices. The sets of natural numbers, real
numbers, non-negative real numbers and positive real numbers are
denoted by $\mathbb{N},\real$, $\real_{\geq0}$ and $\real_{>0}$,
respectively. Let $\boldone \in \realn$ and $\boldzero \in \realn$
denote a vector with all elements equal to one and zero,
respectively. For any vector $\boldz \in \realn$, $\boldz^\top$
denotes its transpose. For any scalar $a \in \real$, $|a|$ denotes its
absolute value. For a set $\Mc{S}$, $\mathcal{S}^{n}$ denotes the
Cartesian product of $\mathcal{S}$ with itself $n$ times. The empty
set is denoted by $\varnothing$. We denote the difference of any two
sets $\mathcal{S}_{1}$ and $\mathcal{S}_{2}$ by
$\mathcal{S}_{1}\setminus\mathcal{S}_{2}$. Let
$\grph\ldef \big(\agt, \edg \big)$ be a directed graph, where $\agt$
is the set of nodes and $\edg$ is the set of directed arcs.  In a
directed graph $\grph$, a \emph{directed walk} from a node
$i_{1} \in \agt$ to any node $i_{l} \in \agt$ is a sequence of nodes
$i_{1}\mapsto i_{2}\mapsto \ldots \mapsto i_{l}$ such that
$(i_s, i_{s+1}) \in \edg, \ \forall s \in \{ 1, 2, \ldots, l-1 \}$.

\subsubsection*{Organization of the paper}

The rest of the paper is organized as
follows. Section~\ref{sec:preliminaries} introduces essential
preliminaries for subsequent analysis.
Section~\ref{sec:modeling_and_problem_setup} presents the model for
utility functions, and the opinion dynamics, as well as outlines the
objectives of the paper. Section~\ref{sec:asymptotic_analysis}
includes the asymptotic analysis of the model.
Section~\ref{sec:eqpt_and_game_theoretic_analysis} analyzes consensus
equilibria, Nash equilibria and price of anarchy of the opinion
formation game. Section~\ref{sec:oscillatory_behaviors} deals with
oscillatory behavior of two-agent opinion dynamics.
Section~\ref{sec:sims} includes simulations demonstrating our
results. Finally, we conclude in Section~\ref{sec:conclusions}.  We
have included some of the longer proofs in the appendix.
   
\section{Preliminaries} \label{sec:preliminaries}

Here, we recall some useful concepts from contraction theory. For a
comprehensive description and a compilation of results on contracting
dynamical systems, we refer the reader to
\cite{2023_FB_CTDS_book}. First, we define the \emph{induced
  logarithmic norm} of a matrix and give an interpretation of the
same.

\begin{define}{(\textbf{\emph{Logarithmic norm of a
        matrix}}\cite{2023_FB_CTDS_book})} \label{def:log_norm} Given
  an induced matrix norm $\|.\|$, the induced log-norm of a matrix
  $\mathbf{A} \in \real^{n \times n}$ is given by
\begin{equation*}
	\mu(\mathbf{A}) \ldef \lim_{h \to 0^{+}} \frac{\|\mathbf{I}_{n}+h\mathbf{A}\|-1}{h} \:\:\in \real.
\end{equation*}
If we use the induced $\infty-$matrix norm in the above definition, then we get the induced $\infty-$log norm of $\mathbf{A}$ as
\begin{equation*}
	\mu_{\infty}(\mathbf{A})= \max_{i \in \agt} \Big(a_{ii} + 
	\sum_{j \in \agt \setminus \{i\}} |a_{ij}|\Big). \tag*{\bulletend}
\end{equation*} 
\end{define}
The log-norm can be interpreted as the directional derivative of the 
matrix norm in the direction of $\mathbf{A}$ evaluated at the 
identity matrix $\mathbf{I}_{n}$.
It should be noted that the induced log-norm of a matrix is not a 
matrix norm and can even be negative.
This induced log-norm helps us in getting bounds on the norm of the 
solutions of a continuous time system. For example, consider a 
continuous-time homogeneous LTI system 
$\dot{\mathbf{x}}=\mathbf{A}\mathbf{x}$. Using Coppel's 
inequalities \cite{1965_coppel}, it can be shown that, 
$$\|\mathbf{x}(t)\| \leq 
e^{{\mu(\mathbf{A})}t} \|\mathbf{x}(0)\|.$$
If $\mu(\mathbf{A})<0$ then the above inequality implies that 
the solution $\mathbf{x}(t)$ converges to zero exponentially. We can 
also extend this idea to nonlinear systems. With this motivation, we 
define a \emph{strongly contracting vector field}. 
\begin{define}{(\textbf{\emph{Strongly contracting vector field}} 
\cite{2023_FB_CTDS_book})}  \label{def:strong_contraction}
	Let $\Mc{C} \subset \realn$ be convex. Let $\mathbf{f}: \Mc{C} \to 
	\realn$ be differentiable and let $\mathbf{J}(\boldz)$ denote its 
	Jacobian. Then the vector field $\mathbf{f}$ is said to be 
	\emph{strongly infinitesimally contracting} on $\Mc{C}$ with rate 
	$\alpha >0$ if 
  \begin{equation*}
    \mu(\mathbf{J}(\boldz))<-\alpha\:,\:\forall \: 
    \boldz \in \Mc{C}. \tag*{\bulletend}
  \end{equation*}
\end{define} 

Finally, we recall a result that ensures the existence and uniqueness
of an exponentially stable equilibrium for a strongly contracting
system. The proof of this theorem is based on the Banach contraction
theorem~\cite{2023_FB_CTDS_book}.

\begin{theorem}{(\textbf{\emph{Equilibrium of a strongly contracting
        system}}\cite{2023_FB_CTDS_book})} \label{thm:contraction}
  Suppose $\Mc{C} \subset \realn$ is convex, closed and positive
  $\mathbf{f}$-invariant. If $\mathbf{f}:\Mc{C}\to\realn$ is strongly
  infinitesimally contracting on $\Mc{C}$ with rate $\alpha>0$ then,
  $\mathbf{f}$ has a unique globally exponentially stable equilibrium
  $\boldz^{*}\in \Mc{C}$
    with global Lyapunov functions $V(\boldz)=\|\boldz - \boldz^{*}\|$
    and $V(\boldz)=\|\mathbf{f}(\boldz)\|$.
  \remend
\end{theorem}

\section{Modeling and Problem Setup}
\label{sec:modeling_and_problem_setup}

Consider a set $\agt:=\{1,\ldots,n\}$ of $n$ agents, with
heterogeneous resources, that form opinions on a single topic, which is about the quantity of resources to use for a certain purpose.
We aim to study the evolution of these opinions as the agents interact
with each other over a social network. We start by defining utility
functions for each agent.
Then, we discuss the utility function and our motivations behind
choosing it.
Then, we derive the proposed opinion dynamics from the
utility functions, assuming that each agent myopically seeks to
maximize its utility. Thus, the coupling in the utility functions
creates the coupling in the opinion dynamics.
	
\subsubsection*{Opinions, Utility Function and its Parameters}

We denote the \emph{expressed opinion} of agent $i \in \agt$ at time
$t$ on the topic as $\zi(t) \in \real$. For brevity, we omit the time
argument wherever there is no confusion. The vector
$\boldz \ldef [z_1,\cdots, z_n]\trpose \in \realn$ represents the
stacked opinions $\zi$ of all agents $i \in \agt$. We first present
the \emph{utility function} of agent $i$, then describe the various
parameters and provide motivation for the chosen structure.

The complete opinion profile $\boldz \in \real^n$ determines the
utility for each agent $i \in \agt$ as follows
  \begin{equation} U_i (\boldz,\di) = -\frac{\wi}{2} \big(\zi - \di
    \big)^{2} -\sum_{k\in \agt} \frac{\aik}{2}\left(\zk - \zi
    \right)^{2} -\frac{1}{4\ri}\zi^{4},
    \label{eq:utility}
  \end{equation}
where $z_i \in \real$ and $p_i \in \real$ are the expressed opinion
and internal \emph{preference} on the topic of agent $i \in \agt$
respectively. The parameter $w_i \in \real_{> 0}$ is the importance
that agent $i \in \agt$ attaches to its internal preference on the
topic, while $r_i \in \real_{>0}$ is the resources available to agent
$i \in \agt$ and $a_{ik} \in \real$ is the weight of the influence of
agent $k \in \agt$ on agent $i \in \agt$.
We call the three terms in the utility function as the
\emph{preference term}, \emph{social term} and \emph{resource
penalty}, respectively.
We make the following standing assumption about the parameters 
in~\eqref{eq:utility}.
\begin{enumerate}[label=\textbf{(SA\arabic*)},wide=\parindent] 
  \item \thmtitle{Parameters' signs}
  For each agent $i \in \agt$, $\di \in \real$, $\wi \in \real_{> 0}$, 
  $\ri \in \real_{>0}$ and $\aik \in \real$, $\forall k \in \agt$. 
  \remend
  \label{asmp:param_sign}
\end{enumerate}

Note that we allow for heterogeneous agents where each agent can have 
different parameters in their utility function. Moreover, in this 
paper, we assume that an agent's utility is affected by others via a 
directed social influence graph $\grph\ldef \big(\agt, \edg, \adjmat 
\big)$.
The elements of the \emph{adjacency matrix} $\adjmat$ are denoted by 
$\aij \in \real$. If $\aik \neq 0$ then there exists a directed link 
from node $k \in \agt$ to node $i \in \agt$ with link weight $\aik$. 
This denotes that agent $k$'s opinion influences the opinion of agent 
$i$. The sign of $\aik$ denotes the type of influence relationship and 
its magnitude denotes the degree of influence. Using this idea, we give 
the following definition.
\begin{define} \label{def:neighbor_set}
  \thmtitle{Neighbor set of an agent}
  For each agent $i \in \agt$, the set $\Nie \ldef \{k \in 
  \agt\setminus\{i\} \:|\: \aik<0 \}$ denotes the set of its 
  \emph{enemies} and the set $\Nif \ldef \{k \in \agt\setminus\{i\} 
  \:|\: 
  \aik > 0\}$ denotes the set of its \emph{friends}. Further,
  \begin{align*} \label{neighbors}
    \Ni \ldef \Nie \cup \Nif = \{ k \in \agt\setminus\{i\} \:|\: \aik 
    \neq 0 \},
  \end{align*}	 
  denotes the set of neighbors of agent $i \in \agt$.	
  \bulletend
\end{define}

Note that the self loop weights $a_{ii}$, for any $i \in \agt$, do not 
affect the utility of the agents. We may assume them to be zero without 
loss of generality. The self influence of the agents is captured by the 
preference term in the utility function. 

\subsubsection*{Discussion About the Utility Function}
  Overall, we seek to explain opinion evolution in social networks
  from a utility maximization perspective. In the sequel, we derive
  opinion dynamics as a gradient ascent by the agents of their utility
  functions, with respect to their own opinions. The three terms
  in~\eqref{eq:utility} represent penalties on agent $i$'s expressed
  opinion $z_i$ arising from three factors.  The preference term
  penalizes opinion deviation from the internal preference $\di$. This
  penalty is directly proportional to the importance weight (or
  stubbornness) $w_i$. The social term penalizes the agent's opinion
  for being far away from or close to its neighbors' opinions
  depending upon the type of influence relationships, i.e., on the
  sign of $a_{ik}$'s. 
  These two terms in the utility function lead to the first two terms
  in the opinion dynamics~\eqref{eq:dynamics}, which is the well known
  Taylor's model of opinion dynamics~\cite{1968_Taylor}, when only
  non-antagonistic relations are allowed. In the presence of antagonistic
  relations, the second term 
  in~\eqref{eq:utility} is motivated by the \emph{boomerang
    effect}~\cite{1967_RA_JM_boomerang_effect}. So, our proposed model
  differs from the literature primarily in the resource penalty in the
  utility function~\eqref{eq:utility}.

  The motivation for the resource penalty comes from the observation
  that if the topic of discussion is about the resource usage for some
  purpose, then the resource limitations of the agents must also have
  an effect on the opinion evolution. In our model, the resource
  penalty in the utility function restricts agent $i$ from holding
  opinions of larger magnitude. In order to better motivate our model,
  consider the following example.

  \begin{example}
    Consider $n$ agents with heterogeneous resource limitations. Let
    $z_i(t)$, the opinion of agent $i \in \agt$ at time $t$, be the
    maximum amount of resources agent $i$ is willing to spend on buying a
    particular product, if it were to do the buying at time $t$. Agent
    $i$ has an internal opinion about the value of the product,
    denoted by $\di$. Further, $\wi$ models the confidence of agent
    $i$ in its internal opinion. In addition, the agents are also
    influenced by the opinions of others within their social
    neighborhood. For example, they might look at product reviews or
    inquire within their social circles, which can influence their
    opinion about the value of the product. The weights $\aik$'s model
    the trust or confidence of agent $i$ has in agent $k$'s ability to
    correctly value the product. Alternatively, these weights could
    simply model the degree to which agent $i$ seeks to mimic or
    disagree with agent $k$.
    However, its spending is ultimately constrained by its
    resource limitations, which motivates the resource penalty
    in~\eqref{eq:utility}.  \remend
  \end{example}

  The psychological and social effects of resource constraints on
  decision making of the agents have been studied in various fields
  such as psychology, economics, transportation etc. We refer the
  readers to the papers~\cite{2019_RH_CM_AS_DT_VG,2015_CH_CM} and
  references therein for more details. For example,
  ~\cite{2019_RH_CM_AS_DT_VG} discusses the effects of financial
  constraints on consumer behaviors from four different perspectives:
  \emph{resource scarcity}, \emph{choice restriction}, \emph{social
    comparison} and \emph{environmental uncertainity}. 

  \begin{remark}
    \thmtitle{On the Resource Penalty} The parameter $\ri$ denotes the
    amount of the resources (such as wealth, time etc.)  available to
    agent $i \in \agt$. It could be a proxy for the maximum budget an
    agent has to spend, such as the amount of money available to buy a
    certain good, the amount of time or other resources available for
    doing a task etc. In our work, we assume that the amount of
    resource $\ri$ is static in time. This is reasonable to assume in
    scenarios where the agents' opinions are about the usage of their
    resources itself. The agents actually use their resources only
    after its opinion converges or after a sufficiently long evolution
    time of its opinion. Thus, we can think of the resource penalty as
    a soft constraint on their opinions, which are in turn about how
    much resources they could expend for a good or a task.  The
    greater the resources that agent $i$ has, the larger the magnitude
    of the opinions it can hold.

    In this work, we choose a quartic resource penalty function;
    however, a more general class of penalty functions is also
    acceptable. Mathematically, all that is needed for ensuring the
    boundedness of opinions is to choose a non-negative resource penalty term
    that dominates the other terms in the utility~\eqref{eq:utility} for
    large enough $| \zi |$. \remend
\end{remark}

The following is a standing assumption in this paper.
\begin{enumerate}[resume,label=\textbf{(SA\arabic*)},wide=\parindent] 
  \item \thmtitle{Preferences are fixed parameters}
  For each agent $i \in \agt$, $\di \in \real$ is a fixed parameter. 
  \remend
  \label{asmp:pi_fixed}
\end{enumerate}

The results can easily be extended for the case where $\di = 
\zi(0)$, for each agent $i \in \agt$. For example, we can let $\boldp = 
\boldz(0)$ and consider $(\boldz,\boldp)$ as the state variables with 
$\dot{\boldp}=0$. Since this provides little additional value, we 
choose to think of $\boldp$ as a fixed parameter. In the case where 
$\boldp = \boldz(0)$, we can study the dependence of the opinion 
evolution on the initial opinions $\boldz(0)$ by carrying out a 
parametric study.
\subsubsection*{Opinion Dynamics}
\par We assume that at each time instant, agent
$i \in \agt$ revises its opinion by doing a gradient ascent of its utility
function $U_i$, given in~\eqref{eq:utility}, with respect to its own
opinion $\zi$. Thus, for each $i \in \agt$, we have
\begin{equation}\label{eq:dynamics}
	\begin{split}
		\dot{\z}_{i} = &-\wi [ \zi-\di ] +
		\sum_{k \in \agt} \aik [\zk - \zi]
		-\frac{\zi^{3}}{\ri}.
	\end{split}
\end{equation}

We can rewrite~\eqref{eq:dynamics} equivalently as
  \begin{subequations}
    \begin{align}
      \label{eq:dynamics_split}
      \dot{\z}_{i} = f_i(\boldz,\di) &:=
      \Sf_{i}(\zi,\di)+\Cf_{i}(\boldz),\:\:\forall i \in \agt, 
      \\
      \label{eq:func_self} \Sf_{i}(\zi,\di) %
                                     &:=-\wi\left[\zi-\di \right]-\frac{\zi^{3}}{\ri}, 
      \\
			\label{eq:func_consensus}
			\Cf_{i}(\boldz) &:= \sum_{k \in \agt}\aik [\zk - \zi].
		\end{align}\label{eq:func_split}
	\end{subequations}
Note that $\forall i \in \agt$, the \emph{self function}
$\Sf_{i}(\cdot,\di)$ depends only on $i$'s own opinion, its
preference $\di$ and other parameters. On the other hand,
$\forall i \in \agt$, the \emph{crowd function} $\Cf_{i}(\cdot)$
depends on the deviations of $i$'s opinion $z_i$ from its neighbors'
opinions $\zk$'s. If $\aik > 0\:(<0)$ then $i \in \agt$ wants to agree
(disagree) with agent $k$ and hence the term in $\Cf_{i}(\cdot)$
corresponding to agent $k$ drives $\z_i$ towards (away from) the
opinion of $k \in \agt$.

Now, we make some observations about the self function 
$\Sf_{i}(\cdot,\di)$, which hold $\forall i \in \agt$. $\Sf_{i}(\cdot,\di)$ is 
continuous and strictly decreasing with $\displaystyle \lim_{\zi 
\rightarrow -\infty} \Sf_i(\zi,\di) = \infty$ and
$\displaystyle \lim_{\zi \rightarrow +\infty} \Sf_i(\zi,\di) = -
\infty$. Thus, $\Sf_{i}(\cdot,\di)$ has exactly one
real root for every fixed value of $\di \in \real$. Let us denote the real root of $\Sf_{i}(\cdot,\di)$ as $m_{i}(\di) 
\in \real$, \emph{i.e.}, $\Sf_{i}(m_i(\di),\di) = 0$. Hereafter, we will exclude the preference argument in $\Sf_{i}(\cdot,\di)$ and $m_{i}(\cdot)$ wherever there is no confusion. Moreover, by considering 
$\Sf_i(0)$ and $\Sf_i(\di)$, we can verify that $0 \leq |m_i| \leq 
|\di|$ and $m_i \di \geq 0$. $m_{i} \in \real$ can be interpreted as 
the opinion that agent $i$ would attain under \eqref{eq:dynamics} if it 
were \emph{socially closed} 
from the influence of other agents, \emph{i.e.} $\aik = 0$, $\forall k 
\in \agt$.

\subsubsection*{Opinion Dynamics in Absence of Antagonistic Relations}

An important special case is one where there are no antagonistic 
relations, which we formally state in the following assumption.
\begin{enumerate}[resume,label=\textbf{(A\arabic*)},wide=\parindent] 
  \item \thmtitle{No antagonistic relations}
  \label{asmp:no_antagonist}
  $\forall i \in \agt$, $\Nie = \varnothing$.  \bulletend
\end{enumerate}

Under Assumption~\ref{asmp:no_antagonist}, opinion dynamics 
\eqref{eq:dynamics_split} reduces to
\begin{align}
  \dot{\z}_{i} = f_i(\boldz,\di) &:= 
  \Sf_{i}(\zi)+\Cf_{i}^{+}(\boldz),\:\:\forall i \in
  \agt, \label{eq:dynamics_no_antagonists}
  \\
  \text{where } \ 
  \Cf_{i}^{+}(\boldz) &=
                          \begin{cases}
			0\:&; \:\:\text{if}\:\: \Ni = \varnothing, \\
			\displaystyle{\sum_{j \in \agt}\aij} \, \big[ \zbari-\zi \big] \:&;\:\: \text{otherwise} 
		\end{cases}
		\\
		\label{zbar}
		\zbari &\ldef  \agtsumk \frac{\aik}{\left(\sum_{j \in 
				\agt}\aij\right)} \zk
		\,.
	\end{align}
Note that, under Assumption~\ref{asmp:no_antagonist}, for any $i \in 
\agt$, if $\Ni \neq \varnothing$ then $\zbari$ 
in~\eqref{zbar} is well defined. Moreover, notice that this 
$\zbari$ is a convex combination of the opinions $\zk$ of agent $i$'s 
neighbors. So in \eqref{eq:dynamics_no_antagonists}, it suffices 
$\forall i \in \agt$ to know only $\zbari$ and be unaware of other 
individual agents' opinions. Finally, note that $\forall i \in \agt$,
	\begin{equation}
		\Sf_i(\zi)
		\begin{cases}
			> 0, \quad & \zi < m_i \\
			= 0, \quad & \zi = m_i \\
			< 0, \quad & \zi > m_i 
		\end{cases},
		\,\,
		\Cf_i^{+}(\boldz)
		\begin{cases}
			> 0, \quad & \zi < \zbari \\
			= 0, \quad & \zi = \zbari \\
			< 0, \quad & \zi > \zbari.
		\end{cases}
		\label{eq:SC_sign}
	\end{equation}

\subsubsection*{Objectives} The important analytical questions regarding the opinion dynamics~\eqref{eq:dynamics} that we study in this paper include
\begin{itemize}
  \item asymptotic properties
  \item existence of a consensus equilibrium and its properties
  \item characterization of Nash equilibria and their relation to the 
  equilibria of the dynamics~\eqref{eq:dynamics}
  \item bounds on price of anarchy
  \item oscillatory behavior of opinions in the two agent case.
\end{itemize}
We also study these questions for the cases outlined in 
Assumption~\ref{asmp:weak_antagonist_relations} 
and~\ref{asmp:no_antagonist} to provide stronger results.

\section{Asymptotic Behavior of Opinions} 
\label{sec:asymptotic_analysis} \label{sec:long_term_behavior}

In this section, we study the asymptotic behavior of opinions. We show 
ultimate boundedness of opinions and provide conditions that guarantee 
existence and uniqueness of a globally exponentially stable equilibrium 
point. We denote the set of equilibrium points of~\eqref{eq:dynamics} as a function of preferences $\mathbf{p}$ as
\begin{equation}
	\eqptp \ldef \{ \boldz \in \real^{n} \,|\, \dot{\boldz} = \mathbf{f}(\boldz,\boldp)= \mathbf{0}
	\}.
	\label{eq:eqpt}
\end{equation}
Hereafter, we drop the preference argument in $\eqptp$ wherever there 
is no confusion. 
We first guarantee that, under the dynamics~\eqref{eq:dynamics}, 
irrespective of the initial opinion profile, the opinions never grow 
unbounded and are in fact, ultimately bounded. The ultimate boundedness 
of opinions is a consequence of the resource penalty term in 
\eqref{eq:utility}, which heavily penalizes any agent for holding opinions of greater magnitude. 

\begin{theorem} \label{thm:ultimate_bound} \thmtitle{Ultimate
    boundedness of opinions} Let $\boldz(t)$ be the solution
  to~\eqref{eq:dynamics} from the initial condition $\boldz(0)$. Then
  $\exists \, \, \eta \geq 0$ (independent of $\boldz(0)$) and
  $\exists \, \, T(\boldz(0)) \geq 0$ such that $|\zi(t)| \leq \eta$,
  $\forall i \in \agt$, $\forall t \geq T$. Additionally,
  $\exists\: \Omega \subseteq \realn$ which is convex, compact and
  positive invariant under~\eqref{eq:dynamics}. \remend
\end{theorem}
  \par The existence of an ultimate bound can be shown using similar
  arguments involved in the proof of~\cite[Theorem
  3.1]{2023_PW_NM_PT}. Using the same
  arguments, we can also show existence of a convex, compact sublevel
  set $\Omega$ of the quadratic Lyapunov function
  $V(\boldz):=0.5\;\boldz\trpose\boldz$ such that $\dot{V}(\boldz)<0$
  outside $\Omega$ and the solution with any initial 
  condition $\boldz(0) \in \real^n$ enters the set $\Omega$ 
  in finite time $\bar{T}(\boldz(0))$. We skip the proof for brevity.

Note that one can explicitly find an ultimate bound $\eta$ but we skip 
it for brevity. Also note that Theorem~\ref{thm:ultimate_bound} holds 
for all values of the parameters and hence does not depend on the type 
of interactions among agents or the structure of the social network. 
From the proof of Theorem~\ref{thm:ultimate_bound}, we note that for 
any initial condition, the solution is uniformly bounded over all time. 
Then the local Lipschitzness of the vector field in~\eqref{eq:dynamics} 
can be used to show existence and uniqueness of solutions 
of~\eqref{eq:dynamics} for all time.
Next, we characterize the conditions under which the opinions converge 
to an equilibrium. Recall that for each $i \in \agt$, 
$\Sf_i(\cdot)$ has a unique root at $m_i$, \emph{i.e.} $\Sf_i(m_i) = 
0$. Using this, we show that the opinion of a socially closed agent (if 
one exists) converges to its corresponding $m_i$. 
The next result can be shown using~\eqref{eq:SC_sign} and the fact that the
dynamics~\eqref{eq:dynamics} for every such agent reduces to the
scalar differential equation $\dot{\zi} =\Sf(\zi)$. We skip the
proof for brevity.

\begin{lemma} \label{lem:eqm_isolated_agents}
\thmtitle{Convergence of opinions of socially closed agents}
Consider the dynamics given by \eqref{eq:dynamics_split}. Let 
$\boldz(t)$ denote the solution of~\eqref{eq:dynamics} from an initial 
condition $\boldz(0)$. 
Then, $\lim_{t \to \infty} \zi(t) = m_{i}\:;\:\forall i \in \agtiso$, 
where $\agtiso$ is the set of socially closed agents, i.e., $\agtiso \ldef \{i \in \agt\:\mid\: \Ni = \varnothing\} \subseteq \agt$.
\remend
\end{lemma}

We next address the general case, where there may be some socially closed agents and some agents influenced by others. In this case, the 
parameters, such as importance weights, resources and 
inter-agent influence weights, all play a role in determining whether 
there is an equilibrium point, its uniqueness and its stability. Using 
Theorem~\ref{thm:contraction}, we present a sufficient condition 
for~\eqref{eq:dynamics} to have a unique globally exponentially stable 
equilibrium. We begin by stating an assumption.
\begin{enumerate}[resume,label=\textbf{(A\arabic*)},wide=\parindent] 
  \item \thmtitle{Weak antagonistic relations}
  \label{asmp:weak_antagonist_relations}
  For each agent $i \in \agt$, $\wi> \enemysumi 2|\aik|$, where $\Nie$ is the set of enemies of $i$.  \bulletend
\end{enumerate}

Notice that under Assumption~\ref{asmp:param_sign}, $\wi>0, \forall i \in \agt$. Thus, the case of no antagonistic 
relations stated in Assumption~\ref{asmp:no_antagonist} is a special 
case of  Assumption~\ref{asmp:weak_antagonist_relations}. All the 
results which hold under 
Assumption~\ref{asmp:weak_antagonist_relations} also hold under 
Assumption~\ref{asmp:no_antagonist}. Now, we present the main result of 
this section. Its proof is in the appendix.
\begin{theorem}  
	\label{thm:exis_uniq_of_eqm}
	\thmtitle{Existence and uniqueness of equilibrium points} 
	Consider the dynamics given by~\eqref{eq:dynamics}. 
	Suppose Assumption~\ref{asmp:weak_antagonist_relations} holds.
Then, there exists a unique globally exponentially stable equilibrium 
point $\boldz^{*} \in \eqpt$. \remend
\end{theorem}

\begin{remark}\thmtitle{On the weak antagonistic relationships condition}
	\label{rem:uniq_eqm}
	The condition of Assumption~\ref{asmp:weak_antagonist_relations} 
	means that the \emph{willingness} of every agent $i \in \agt$ to hold 
	an opinion close to its internal preference is greater than twice the 
	aggregate influence weight of its enemies. Under this assumption, 
	Theorem~\ref{thm:exis_uniq_of_eqm} shows that the opinions converge 
	to a unique equilibrium.
	\remend
\end{remark}

Based on Lemma~\ref{lem:eqm_isolated_agents} and 
Theorem~\ref{thm:exis_uniq_of_eqm}, we can give a stronger 
result for ultimate boundedness and on the location of the unique 
equilibrium in case there are no antagonistic relations among agents. 
We deal with this case next.

\subsubsection*{Convergence of opinions in the absence of antagonistic 
relations} 

In the absence of antagonistic relationships among agents, 
it is possible to give an ultimate bound that is more intuitive and 
easier to compute using the unique roots of the self functions 
($m_i$ such that $\Sf_i(m_i) = 0$, $i \in \agt$). We can also say that 
the unique equilibrium in this case is in a specific set defined by 
$m_i$'s, for $i \in \agt$. Let,
\begin{align} 
	m_{\min}(\boldp):=\min \{m_{i}\}_{i \in \agt}, \, m_{\max}(\boldp):=\max \{m_{i}\}_{i \in \agt},
	\label{eq:m_min_max}
\end{align}
and the corresponding interval
\begin{align} 
	\Mc{M}(\boldp):=[m_{\min}(\boldp),m_{\max}(\boldp)]\,.
	\label{eq:m_interval}
\end{align}
 For brevity, we will exclude the preference argument in~\eqref{eq:m_min_max} and~\eqref{eq:m_interval} wherever there is no confusion. We are now ready to show that the opinions converge to~$\Mc{M}^n$ (proof in Appendix).
 
\begin{proposition}\label{prop:ub_c_pos}
	\thmtitle{Convergence to the set $\Mc{M}^n$ in absence of antagonistic relations} \label{prop:conv_to_compact_set}
	Consider the opinion dynamics given by \eqref{eq:dynamics}. 
	Suppose Assumption~\ref{asmp:no_antagonist} holds. Let $m_{\min}$, 
	$m_{\max}$ and $\Mc{M}$ be as defined 
	in~\eqref{eq:m_min_max} and~\eqref{eq:m_interval} respectively. Then, 
	$\Mc{M}^n$ is positively invariant under the opinion 
	dynamics~\eqref{eq:dynamics}. Let $\boldz(t)$ be the 
	solution to \eqref{eq:dynamics} from an initial 
	condition $\boldz(0) \in \real^n$. Then $\boldz(t)$ converges to 
	$\Mc{M}^n$. Further, define
	$\agtmmax \ldef \{i \in \agt \mid m_{i}=m_{\max}\}$ and $\agtmmin \ldef \{i \in \agt \mid m_{i}=m_{\min}\}$. Suppose $m_{\min} < m_{\max}$. Then the following statements are equivalent.
	\begin{itemize}[wide=\parindent]
		\item[(i)] The unique equilibrium $\boldz^* \in \eqpt$ lies in the 
		interior of $\Mc{M}^n$ and $\exists \, T(\boldz(0)) \geq 0$ such 
		that $\boldz(t) \in 
		\Mc{M}^n$, $\forall t \geq T(\boldz(0))$.
		\item[(ii)] $\forall i \in \agtmmax$, $\exists$ a directed walk in 
		$\mathsf{G}$ starting from $j \in \agt \setminus \agtmmax$ to $i$ 
		and $\forall i \in \agtmmin$, $\exists$ a directed walk in 
		$\mathsf{G}$ starting from $j \in \agt \setminus \agtmmin$ to 
		$i$.~\remend
		\end{itemize}
\end{proposition}  

Under Assumption~\ref{asmp:no_antagonist}, 
Proposition~\ref{prop:ub_c_pos} guarantees that the set $\Mc{M}^n$ is 
positively invariant. Moreover, the unique globally exponentially 
stable equilibrium point $\boldz^*$ that 
Theorem~\ref{thm:exis_uniq_of_eqm} guarantees lies in $\Mc{M}^n$. Thus, 
we can immediately guarantee that $\boldz(t)$ converges to $\Mc{M}^n$, 
possibly asymptotically. If $\mmin<\mmax$ then only agents belonging to 
sets $\agtmmax$ and $\agtmmin$ can have equilibrium opinions at the 
boundary of the set $\Mc{M}$. For the unique equilibrium to be in the 
interior of $\Mc{M}^{n}$, it is both necessary and sufficient that 
every agent in the sets $\agtmmax$ and $\agtmmin$ is directly or 
indirectly influenced by at least one agent in sets $\agt \setminus 
\agtmmax$ and $\agt \setminus \agtmmin$, respectively. Note that this 
condition is satisfied if the social network $\mathsf{G}$ is strongly 
connected. Finally, if the unique equilibrium $\boldz^*$ lies in the 
interior of $\Mc{M}^n$ then solutions converge to $\Mc{M}^n$ in finite 
time. In this case, the ultimate bound $\Mc{M}^n$ has an additional 
advantage that it depends only on the $m_i$'s (whose interpretation is 
provided in Section~\ref{sec:modeling_and_problem_setup}) and hence can 
be computed easily using the parameters $w_i, p_i$ and $r_i$. Finally, 
note that if $m_{\max}=m_{\min}$ then the opinions of agents in absence 
of antagonistic relations always converge to a unique consensus 
equilibrium.

\section{Consensus and Nash Equilibria} 
\label{sec:eqpt_and_game_theoretic_analysis}

In this section, we analyze consensus equilibria of the opinion
dynamics~\eqref{eq:dynamics} and Nash equilibria of the underlying
game. We also explore the relation between the Nash equilibrium set of
the underlying game and the set of equilibria of the opinion
dynamics~\eqref{eq:dynamics}. Finally, we also analyze the price of
anarchy.

\subsection{Consensus Equilibria}

First, we deal with the consensus equilibria of the model, i.e.,
equilibria of the form $\xi \one$, with $\xi \in \real$. We refer to
the case of $\xi = 0$ as a \emph{neutral consensus} since all the 
agents have neutral opinions (equal to $0$) in this case. On the other 
hand, we refer to the case of $\xi \neq 0$ as a non-neutral consensus. 
In the following lemma, we present conditions for \eqref{eq:dynamics} 
to have a consensus equilibrium. We use the form of the dynamics in
\eqref{eq:dynamics_split} and the functions in \eqref{eq:func_split}
to justify our claims.

\begin{theorem}\label{thm:consensus_iff}
	\thmtitle{Necessary and sufficient conditions for existence of a
		consensus equilibrium} 
	Consider the dynamics \eqref{eq:dynamics} (equivalently~\eqref{eq:dynamics_split}). For each $i \in \agt$, let
	$m_i \in \real$ be the unique point such that $\Sf_i(m_i) =
	0$. Then, $\xi \one \in \eqpt$ if and only if $m_i = \xi$,
	$\forall i \in \agt$.	
\end{theorem}

\begin{proof}
	If $\z_i = \xi$, $\forall i \in \agt$, for some
	$\xi \in \real$ then $\Cf_i(\xi\one) = 0$, $\forall i \in
	\agt$. Hence, from \eqref{eq:dynamics_split}, $\xi\one \in \eqpt$
	iff $\Sf_i(\xi) = 0, \forall i \in \agt$. Since
	$m_i$ is the unique root of $\Sf_i(.)$, the claim follows.
\end{proof}
\begin{remark}\label{rem:consensus}
	\thmtitle{Consensus formation among agents} %
	Theorem~\ref{thm:consensus_iff} states that it is both necessary and
	sufficient for all the $m_i$'s to be the same for the opinion dynamics
	model to have a consensus equilibrium. It is evident that if the
	agents are to arrive at a consensus equilibrium, then all their
	preferences $\di$'s must be of the same sign. When $\di=0$, $\forall i \in \agt$, then the only possible consensus equilibrium is the neutral consensus, \emph{i.e.}, every agent
	reaches a neutral opinion on the topic. If the preferences of the agents
	have different signs, then the opinions of agents can never reach an
	exact consensus in equilibrium. However, other equilibria that are
	arbitrarily close to consensus may still exist.
	
	Now, if the weak antagonistic relationships condition given in Assumption~\ref{asmp:weak_antagonist_relations} holds and there exists a consensus equilibrium,
	then, from Theorem~\ref{thm:exis_uniq_of_eqm}, it is the only 
	equilibrium of the dynamics and
	the agents always achieve consensus starting from any initial opinion vector. 
	 \remend
\end{remark}

When the agents attain a consensus equilibrium, we can measure how
much influence an agent has on the whole group by measuring the
deviation of the consensus value from its preference. Note that if
$p_i = 0$, for some $i \in \agt$, then $m_i = 0$ and hence the only
consensus equilibrium possible (if it exists) is neutral. So we
consider $p_i \neq 0$, $\forall i \in \agt$ to give the next result on
dominance and discuss it in the remark following it.

\begin{proposition} \label{prop:dominance} \thmtitle{Consensus
    deviation from preference} %
  Consider the dynamics \eqref{eq:dynamics} or equivalently
  \eqref{eq:dynamics_split}. Suppose that $p_i \neq 0$,
  $\forall i \in \agt$. For each agent $i \in \agt$, let us define
  $\sigma_{i}\ldef \wi \ri$ and
    $\Delta_{i}(\xi)\ldef|\di-\xi|$. If $\xi\one \in \eqpt$, with
  $\xi \in \real$, then
  $\sigma_{i} \Delta_{i}(\xi) = \sigma_{j} \Delta_{j}(\xi)$,
  $\forall i,j \in \agt$. In particular, $\sigma_{i}>\sigma_{j}$
  iff $\Delta_{i}(\xi)<\Delta_{j}(\xi)$.
\end{proposition}
\begin{proof}
	Since $p_i \neq 0$, $\forall i \in \agt$, we also have
	$0 < | m_i | < | p_i |$ $\forall i \in \agt$. Then by
	Theorem~\ref{thm:consensus_iff}, $\xi \neq 0$. Further, from
	Theorem~\ref{thm:consensus_iff}, we get
	$r_i \Sf_{i}(\xi)=0, \ \forall i \in \agt$, which then implies
	\begin{equation*}
		\sigma_{i}(\di-\xi) = \xi^3 = \sigma_{j}(p_j-\xi), \ \forall i,j \in \agt.
	\end{equation*}
	Since $\sigma_{i}>0, \ \forall i \in \agt$, the result now follows.
\end{proof}

\begin{remark}\label{rem:dominance}
	\thmtitle{Dominance in consensus} %
	Let us call $\forall i \in \agt$, the scalar $\sigma_{i}$ 
	as the \emph{consensus dominance} weight of the agent
	$i$.  Suppose all agents have a non-neutral preference and they
	attain consensus at $\xi \in \real$. Then, Proposition~\ref{prop:dominance} states
	that if an agent $i \in \agt$ has higher consensus dominance weight 
	than agent $j \in \agt$, then $\xi$ is closer to $i$'s preference
	than that of $j$. Note that the consensus dominance weight is directly
	proportional to the importance weight an agent assigns to its preference and the resources available to it. Thus, an agent with very 
	high resources can exert more influence even if it gives less importance weight 
	to its internal preference. On the other hand, if an agent has lower 
	resources, then it has to have much higher importance weight to have 
	more influence in the group.  \remend
\end{remark}

Next, we study the relation between equilibria of
the opinion dynamics~\eqref{eq:dynamics} and Nash equilibria of the 
underlying game.

\subsection{Nash Equilibria}

Here we carry out a Nash Equilibrium analysis of the opinion formation
game. Recall that every agent is interested in maximizing its utility
$U_i$ given in \eqref{eq:utility} by suitably choosing its opinion
$z_i$. Thus, this gives a \emph{strategic form game}
$\mathcal{G}=\langle \agt, \left(\real\right)_{i \in \agt}, 
\left(U_{i}\right)_{i \in \agt}\rangle$ 
among the set of agents
$\agt$, with the \emph{strategy} of agent $i \in \agt$ being its opinion
$z_i \in \real$ and its utility function being $U_i(\cdot)$. For the
sake of convenience, we let $\znoti$ denote the opinions of all agents
other than $i$. Then, the set of \emph{Nash equilibria} of the game
$\mathcal{G}$ is
\begin{align}
	\notag \nep & \ldef \{ \boldz^* \in \real^n \,\,|\,\, \forall i \in \agt, \\
	& \quad U_i(z^*_i,\znoti^*,\di) \geq U_i(z_i,\znoti^*,\di), \forall z_i \in \real \}\,. \label{eq:ne}
\end{align}
Note that for a Nash equilibrium $\boldz^*$, $z^*_i$ is agent $i$'s
best response over all opinions $z_i \in \real$ to $\znoti^*$, the
opinion profile of all the other agents. However, in the
dynamics~\eqref{eq:dynamics}, each agent updates its opinion according
to the gradient ascent of its utility with respect to its opinion while
assuming that the other agents do not change their opinions. Hence,
the agents at each time instant revise their opinion only towards the
``local''
best response. This motivates the next definition.

\begin{define}\label{def:lne}
	\thmtitle{Local Nash equilibrium}
	A strategy profile $\boldz^* \in \real^n$ is said to be a local Nash equilibrium if and only if $\forall i \in \agt$, $\exists \,\, \rho_i \in \real_{>0}$ such that
	\begin{align*}
		U_i(z^*_i,\znoti^*,\di) \geq U_i(z_i,\znoti^*,\di), \,\, \forall\, z_i \,\,\mathrm{s.t.}\,\, |z_i^* - z_i| \leq \rho_i\,.
	\end{align*}
	The set of local Nash equilibria of $\mathcal{G}$ is denoted by $\lnep$. 
	\bulletend
\end{define}

 For simplicity, we will exclude the preference arguments in $\nep$ and $\lnep$ wherever there is no confusion.
It is easy to see that a Nash equilibrium is also a local 
Nash equilibrium and hence $\ne \subseteq \lne$. In the next 
result, we show that every local Nash equilibrium is an 
equilibrium point of~\eqref{eq:dynamics}.
\begin{lemma}\label{lem:lne_subset_eqpt}
  \thmtitle{Local Nash equilibrium is an equilibrium of the opinion
    dynamics} If an opinion profile $\boldz^{*}$ is such that
  $\boldz^* \in \lne$ then $\boldz^* \in \eqpt$.
\end{lemma}
\begin{proof}
	Since $\boldz^* \in \lne$, it implies that $\forall i \in 
	\agt$, $z_i^*$ locally maximizes $U_i(\cdot,\znoti^*)$. 
	Thus, the partial derivative of $U_i(\cdot)$ with respect 
	to $z_i$ evaluated at $(z_i^*,\znoti^{*})$ is \emph{zero}. The claim 
	then follows immediately from \eqref{eq:dynamics} and 
	\eqref{eq:eqpt}.
\end{proof} 

Lemma~\ref{lem:lne_subset_eqpt} states that every local Nash 
equilibrium of $\mathcal{G}$ is also an equilibrium point of 
dynamics \eqref{eq:dynamics}. But the converse need not be 
true. In the following result, we give conditions for an 
opinion profile $\boldz^{*} \in \realn$ that is an 
equilibrium point of the opinion dynamics to be a local 
Nash equilibrium of the opinion formation game.

\begin{theorem}\label{thm:local_nash_eq}
	\thmtitle{Conditions for an equilibrium point of \eqref{eq:dynamics}
		to be a local Nash equilibrium} %
	Consider the dynamics \eqref{eq:dynamics} and the set of equilibrium
	points $\eqpt$ in \eqref{eq:eqpt}. 
	Let $\boldz^{*}=\left(\zi^{*},\znoti^{*}\right) \in \eqpt$. Then, 
	$\boldz^* \in \lne$ only if for each $i \in \agt$,
      \begin{equation}
        (z_i^*)^2 \geq \tau_i := \frac{\ri}{3} 
        \left[\enemysumi |\aik| - \friendsumi \aik - \wi\right].
        \label{eq:lne_inequality}
      \end{equation}
	Moreover, if inequality~\eqref{eq:lne_inequality} is strict $\forall i \in \agt$, 
	then $\boldz^* \in \lne$.
\end{theorem}
	\begin{proof}
	Let the hypothesis be true. From Definition~\ref{def:lne} we know 
	that $\boldz^{*}=\left(\zi^{*},\znoti^{*}\right)$ is a local Nash 
	equilibrium if and only if $\forall i \in \agt$, $z_i^*$ locally 
	maximizes $U_i(\cdot,\znoti^*)$. Now since $\boldz^* \in \eqpt$, by 
	the definitions in \eqref{eq:dynamics} and \eqref{eq:eqpt}, it is 
	clear that for each $i \in \agt$, $z_i^*$ satisfies the first-order 
	necessary conditions for it to be a local maximizer of 
	$U_i(\cdot,\znoti^*)$.
	
	Now, suppose that $\boldz^* \in \lne$. Then, we have that
	\begin{align}
		\left.\frac{\partial^{2}}{\partial\zi^{2}}U_{i}(\zi,\znoti^{*})\right|_{\zi^{*}}
		  & = \frac{3}{r_i} \left[ \tau_i - (z_i^*)^2 \right] 
		  \leq 0,\quad \forall i \in \agt\,.
		\label{eq:second_derivative}
	\end{align}
	This proves the necessary condition in \eqref{eq:lne_inequality}. 
	Finally, note that if~\eqref{eq:second_derivative} is a strict 
	inequality for a $\boldz^* \in \eqpt$, then by the second order 
	sufficiency condition for a point to be a local maximizer, $z_i^*$ is a local
	maximizer of $U_{i}(\cdot,\znoti^{*})$ for each $i \in \agt$.
\end{proof}

The statement of the previous result can be combined with the result in Theorem~\ref{thm:ultimate_bound} to provide a condition for which the opinion formation game does not have any local Nash equilibrium. We state this in the next result, proof of which is intuitive since no equilibrium point of \eqref{eq:dynamics} can exist beyond any ultimate bound (which always exists). 
\begin{corollary}\label{cor:non_existence_of_nash_eq}
	\thmtitle{Non-existence of local Nash equilibria} %
	Suppose $\eta_{i}>0$ is an ultimate bound on $\zi$ for any $i \in 
	\agt$ under \eqref{eq:dynamics}. If there exists $i \in \agt$ such 
	that
	$\tau_{i}>\eta_{i}^{2}$, with $\tau_{i}$ defined in 
	\eqref{eq:lne_inequality}, then
	$\lne = \varnothing$. \proofend
\end{corollary}

Finally, to end this subsection, we provide a sufficient condition 
under which the different equilibria sets ($\eqpt$, $\ne$ and $\lne$) 
are equal. Additionally, we also give a sufficient condition under 
which the opinion formation game has a unique Nash equilibrium. We 
state this 
in the following result.

\begin{theorem}\label{thm:ne_nc}
	\thmtitle{Equality of equilibria sets and uniqueness of Nash equilibrium}
	Consider the dynamics in \eqref{eq:dynamics_split} and the set of 
	equilibrium points $\eqpt$ in \eqref{eq:eqpt}. Suppose for each agent $i \in \agt$, $\tau_i \leq 0$. Then, $\ne= 
	\lne = \eqpt$. Moreover, if Assumption~\ref{asmp:weak_antagonist_relations} holds then it is a singleton set.
\end{theorem}

\begin{proof} Suppose that $\tau_i \leq 0$, $\forall i \in \agt$.
	First, it is obvious that $\ne \subseteq \lne$ and 
	Lemma~\ref{lem:lne_subset_eqpt} implies that $\lne \subseteq \eqpt$. Next, we show that $\eqpt \subseteq \lne$. Consider an equilibrium $\boldz^{*} \in \eqpt$.
	It can be easily verified that for each $i \in \agt$, $U_{i}(\cdot,\znoti^{*})$ is strictly concave in $\zi$. Since $\boldz^{*} \in \eqpt$, the first order necessary condition for a point to be a local maximizer of $U_{i}(\cdot,\znoti^{*})$ is satisfied for each $i$. Hence, $\boldz^{*} \in \lne$. Finally, we show that $\lne \subseteq \ne$.
  Since $\forall i \in \agt$, $U_{i}(\cdot,\znoti^{*})$ is a 
  strictly concave function for each $\znoti^*$, the implication follows directly from the  definitions of $\ne$ 
  and $\lne$. This completes the proof of the first part of the claim. Now suppose that Assumption~\ref{asmp:weak_antagonist_relations} holds. Then again $\tau_i<0$, $\forall i \in \agt$. Uniqueness of Nash equilibrium now follows from Theorem~\ref{thm:exis_uniq_of_eqm} and first part of the claim.  
\end{proof}
\begin{remark}
  We can interpret
  $\wi$ as the weight of influence of a stubborn virtual agent on
  agent $i$. Using this interpretation, let us define the
  \emph{weighted in-degree} for each agent $j \in \agt$ in social
   network $\mathsf{G}$ as $d_{j}^{in}\ldef 
   \left[w_j + \agtsumknotj a_{jk}\right]$.
   Then 
   for each $j \in \agt$, $\tau_{j} = -\rj d_{j}^{in}/3$ and $\tau_j < 
   0(= 0)$ if and only if $d_{j}^{in} > 0 (= 0)$. The condition 
   $d_{j}^{in}>0$, $\forall j \in \agt$ can be interpreted as follows: 
   For each agent $i \in \agt$, the aggregate influence weight of its 
   friends 
   (including the virtual fully stubborn agent since $\wi > 0$) is 
   greater than the aggregate influence weight of its enemies. If this 
   holds true then Theorem~\ref{thm:ne_nc} states that each equilibrium 
   point of the dynamics is also a Nash equilibrium of the game. Next, 
   if Assumption~\ref{asmp:weak_antagonist_relations} holds then for 
   each $j \in \agt$, $d_j^{in}>0$ and hence $\tau_j<0$. In other 
   words, the opinion formation game has a unique 
	Nash equilibrium if the \emph{willingness} of each agent to be close to its internal preference (or stubbornness) is more than 
	twice the aggregate influence weight of its enemies. Moreover, thanks 
	to Theorem~\ref{thm:exis_uniq_of_eqm}, we can say that the 
	opinions under \eqref{eq:dynamics} always converge to this unique 
	Nash equilibrium, starting from any initial condition. \remend
\end{remark}

\subsection{Price of Anarchy}
Next, we analyze the \emph{price of anarchy} of the game underlying the 
opinion dynamics~\eqref{eq:dynamics}. In the entirety of this 
subsection, we will use cost minimization perspective rather than 
utility maximization. We let the cost incurred by agent $i \in \agt$ 
for opinion profile $\boldz$ be $\costi(\boldz,\boldp) \ldef -U_{i}(\boldz,\boldp)$, 
where $U_{i}(\boldz,\boldp)$ is given by~\eqref{eq:utility}. In order 
to ensure non-negativity of the prices of anarchy defined below, we 
need $\costi(\boldz,\boldp)\geq 0\:;\:\forall \boldz \in \realn$. 
Hence, in this subsection we will assume that all inter-agent relations 
are non-antagonistic, i.e., Assumption~\ref{asmp:no_antagonist} holds. 
We consider price of anarchy for two most commonly used objective 
functions in the game theory literature, namely, the \emph{egalitarian} 
and the \emph{utilitarian} costs
\begin{equation*}
  \Ce(\boldz,\boldp) := \max_{i \in \agt} \costi(\boldz,\boldp), \quad 
  \Cu(\boldz,\boldp) := \sum_{i \in \agt} \costi(\boldz,\boldp),
  \label{eq:utilitarian_cost}
\end{equation*}
to measure the inefficiency of Nash equilibria of the game underlying
the opinion dynamics~\eqref{eq:dynamics}. Hereafter, we will again
exclude the preference arguments in $\costi(\cdot,\cdot)$,
$\Ce(\cdot,\cdot)$ and $\Cu(\cdot,\cdot)$ for brevity. In the
following definitions of price of anarchy, we will assume
$\boldp \neq \boldzero$ to ensure the positivity of $\Ce(\boldz)$ and
$\Cu(\boldz)$ at any $\boldz \in \realn$. This ensures that prices of
anarchy~\eqref{eq:poae} are well defined. We discuss the case when
$\boldp =\boldzero$ in a later remark.
\begin{define}\thmtitle{Price of anarchy}  
	\label{def:poa}
	Consider the opinion formation game $\mathcal{G}$ corresponding to 
	the opinion dynamics~\eqref{eq:dynamics}, and $\ne$, the set of its 
	Nash equilibria. Suppose Assumption~\ref{asmp:no_antagonist} holds 
	and $\boldp \neq \boldzero$. The egalitarian and the utilitarian 
	prices of anarchy, $\poae$ and $\poau$, respectively are defined as
	\begin{equation} \label{eq:poae}
		\poae:= \frac{\max\limits_{\boldz \in \ne }\Ce(\boldz)}{\min\limits_{\boldz \in \realn }\Ce(\boldz)} \geq 1; \,\,\,\, \poau:= \frac{\max\limits_{\boldz \in \ne }\Cu(\boldz)}{\min\limits_{\boldz \in \realn }\Cu(\boldz)} \geq 1.
	\end{equation}
\bulletend
\end{define}
In Definition~\ref{def:poa}, $\poae$ compares the cost incurred by the worst 
performing agent at the worst Nash equilibrium to the minimum 
possible cost $\Ce$. Similarly, $\poau$ compares the total cost 
incurred by all agents at the worst Nash equilibrium to the minimum 
possible cost $\Cu$. The closer the value of the prices of anarchy~\eqref{eq:poae} is to unity, the better the quality of the Nash equilibrium. We now define the \emph{satisfaction ratio} at 
opinion profile $\boldz$ for every agent $i \in \agt$ in order to 
compare the cost incurred by agent $i$ at $\boldz$ and the minimum 
possible cost it could incur. Using these ratios, we can give some 
upper bounds on the prices of anarchy \eqref{eq:poae}. In order to keep the satisfaction ratios well defined, we need positivity of $\costi(\boldz)$ everywhere. Hence we assume the following,
\begin{enumerate}[resume,label=\textbf{(A\arabic*)},wide=\parindent] 
	\item 	\thmtitle{Non-zero preferences}
	$\forall i \in \agt$, $\di \neq 0$.
	\label{asmp:all_agents_have_non_zero_preferences}
	\bulletend
\end{enumerate}
\begin{define}
  \thmtitle{Satisfaction ratio} \label{def:satisfactionratio}
  Consider the cost function $\costi(\boldz)$ of each agent $i \in 
  \agt$. Suppose Assumptions~\ref{asmp:no_antagonist} 
  and~\ref{asmp:all_agents_have_non_zero_preferences} hold. The 
  satisfaction ratio $\satratioi$ for any agent $i \in \agt$ 
  at any $\boldz \in \realn$ is defined as,
  \begin{equation}
    \satratioi(\boldz) \ldef \frac{\costi(\boldz)}{\min_{\boldz \in 
   \realn} \costi(\boldz)} \geq 1.
   \label{eq:satratio}
  \end{equation}
\end{define}
\vspace{3pt}
The next result shows that if there are no antagonistic 
relations among the agents, then for 
each agent $ i \in \agt$, $\costi(\cdot)$ is convex. Note that this 
result says $\costi(\boldz)$ is convex in $\boldz$ and not just with 
respect to $z_i$. Proof of the following result is in the appendix.

\begin{lemma}
	\thmtitle{Convexity of cost function} \label{lem:convexcost}
Suppose that Assumption~\ref{asmp:no_antagonist} holds. Then for each $i \in \agt$,
$\costi(\cdot)$ is convex with $m_{i}\boldone$ as one of its minimizer. 
\remend
\end{lemma}

Now, from 
Theorem~\ref{thm:exis_uniq_of_eqm} and Theorem~\ref{thm:ne_nc}, we 
know that in the absence of antagonistic relations, the game has a 
unique Nash equilibrium $\boldz^{*}\in \ne$. In the next result,  whose 
proof is in the Appendix, we use the satisfaction ratios to give an 
upper bound on $\poae$ and $\poau$.
\begin{theorem}
	\thmtitle{Bounds on prices of anarchy}\label{thm:poa_bound}
	Consider the dynamics \eqref{eq:dynamics}. Suppose 
	Assumptions~\ref{asmp:no_antagonist} 
	and~\ref{asmp:all_agents_have_non_zero_preferences} hold. Let us 
	denote the unique Nash equilibrium of the game underlying the 
  opinion dynamics~\eqref{eq:dynamics} by $\boldz^{*} \in \ne$. Then,
	\begin{equation} \label{ineq:poa_bound}
		1 \leq \poae \leq \max_{i \in \agt} \: \satratioi(\boldz^{*})\:\:, 
		\ 1 \leq \poau \leq \sum_{i \in \agt} \: \satratioi(\boldz^{*})\:\:,
	\end{equation}
 with $\poae$ and $\poau$ as defined in Definition~\ref{def:poa}.
 \bulletend
\end{theorem}

A consequence of this result is that, for the special case where the 
unique Nash equilibrium is a consensus equilibrium, $\poae =\poau=1$, 
i.e., consensus equilibrium (if it exists) is socially optimal. We 
formally state and prove this next.

\begin{corollary}
	\thmtitle{$\poa$ is unity for a non-neutral consensus equilibrium}
	\label{cor:unity_poa}
	Consider the dynamics \eqref{eq:dynamics}. Suppose 
	Assumptions~\ref{asmp:no_antagonist} 
	and~\ref{asmp:all_agents_have_non_zero_preferences} hold. Suppose the 
	unique $\boldz^* \in \ne = \eqpt$ is a non-neutral consensus 
	equilibrium, i.e., $\boldz^* = m\boldone$ for some $m \neq 0$. Then, 
	$ \poae=\poau=1$.
\end{corollary}
\begin{proof}
	 From Theorem~\ref{thm:consensus_iff}, Lemma~\ref{lem:convexcost} and 
	 Definition~\ref{def:satisfactionratio}, it follows that, 
	 $\satratioi(m\boldone)=1,\: \forall i \in \agt$. 
	 From~\eqref{ineq:poa_bound}, we thus have $\poae =1$. From the proof 
	 of Theorem~\ref{thm:poa_bound} for the bound on $\poau$, we have
	 \begin{equation*}
	 	1\leq \poau \leq \frac{\sum_{i \in \agt}\costi(m\boldone)}{\sum_{i \in \agt}\costi(m\boldone)}=1 \,,
	 \end{equation*}
since $\boldz^{*}=m\boldone$ and $m_{i}=m, \:\forall i \in \agt$.
\end{proof}

 Note that neutral and non-neutral consensus equilibrium exists only if 
 $\boldp=\boldzero$ and 
 Assumption~\ref{asmp:all_agents_have_non_zero_preferences} holds, 
 respectively. Consensus equilibrium cannot exist if $ \exists i,j \in 
 \agt$ with $p_{i}=0$ and $p_{j}\neq 0$. Corollary~\ref{cor:unity_poa} 
 states that non-neutral consensus is socially optimal. In the next 
 remark, we argue that neutral consensus is also socially optimal.
\begin{remark}
  \thmtitle{Neutral consensus is socially optimal} Let
  Assumption~\ref{asmp:no_antagonist} hold and $\boldp = \boldzero$.
  Theorems~\ref{thm:consensus_iff},~\ref{thm:exis_uniq_of_eqm}
  and~\ref{thm:ne_nc} imply that the unique equilibrium point
  $\boldz^{*}=\boldzero \in \eqpt = \ne$. Now, from
  Lemma~\ref{lem:convexcost} we know that $\boldzero$ is a minimizer
  of $\costi(\boldz)$, $\Ce(\boldz)$ and $\Cu(\boldz)$. Thus, prices
  of anarchy~\eqref{eq:poae} and satisfaction
  ratios~\eqref{eq:satratio} are not defined in this case. But since
  $\Ce(\boldz^{*})= \min_{\boldz \in \realn} \Ce(\boldz) = 0 =
  \Cu(\boldz^{*}) = \min_{\boldz \in \realn}\Cu(\boldz)$, neutral
  consensus equilibrium is also socially optimal. \remend
\end{remark}

\section{Oscillatory Behavior of Opinions} \label{sec:oscillatory_behaviors}
In~\cite{2023_PW_NM_PT}, we studied a special case of the 
dynamics proposed in the current paper. We showed through numerical 
simulations that in certain scenarios, the opinions exhibit 
oscillatory behavior. In this section, we analyze such behavior in the 
case of a pair of agents $\agt = \{1,2\}$. For a general social network 
with $n$ agents, the analysis is significantly more complicated due to 
the higher dimension of the state space but also due to the complexity 
added by the social network. Thus, the analysis in the general case of 
a social network with $n$ agents would have to be a separate research 
work in itself and is out of the scope of this paper. 
Even though our
analysis is restricted to two agents, it is still relevant since it can help us understand opinion
behaviors or decision making in important systems such as two-party
politics, duopoly economic markets etc.

The opinion dynamics~\eqref{eq:dynamics} for two agents is,
  \begin{equation}
    \dot{z}_{1} = \Sf_{1}(z_{1})+ a_{12} \big[z_{2}-z_{1} \big], \ 
    \dot{z}_{2} = \Sf_{2}(z_{2})-a_{21} \big[z_{2}-z_{1} \big].
    \label{eq:two_agent_dynamics}
  \end{equation}
 For the \emph{two-agent
  dynamics}~\eqref{eq:two_agent_dynamics}, it can be easily verified
that if the opinions $z_{1}(t)$ and $z_{2}(t)$ exhibit oscillatory
behavior then they will have the same fundamental period of
oscillations. We state this in the following result (proof in
Appendix).
\begin{lemma}{(\textbf{\emph{Equal period of oscillations}})} \label{lem:equal_periods}
	Consider the two-agent dynamics \eqref{eq:two_agent_dynamics}. 
	Suppose the opinions $z_{1}(t)$ and $z_{2}(t)$ exhibit oscillations 
	with fundamental periods $T_{1}>0$ and $T_{2}>0$, respectively. Then, 
	$T_{1}=T_{2}$.
	\bulletend
\end{lemma}

Note that periodic orbits can exist for~\eqref{eq:two_agent_dynamics} 
only if at least one of the agents has an antagonist influence on the other as otherwise
Theorem~\ref{thm:exis_uniq_of_eqm} guarantees convergence of opinions 
to the unique equilibrium from any arbitrary initial condition. In the 
next result (proof in Appendix), we give a necessary condition for the 
existence of periodic orbits for~\eqref{eq:two_agent_dynamics}.

\begin{lemma}{(\textbf{\emph{Necessary condition for existence of periodic solutions}})} \label{lem:necessary_condition_for_periodic_orbits}
Consider the two-agent dynamics \eqref{eq:two_agent_dynamics}. 
Opinions 
$z_{1}(t)$ and $z_{2}(t)$ exhibit periodic behavior only if all the following hold: (i) $a_{12} < 0$ or $a_{21} < 0$, (ii) $a_{12}a_{21}\neq 0$, and (iii) $(a_{21}+ w_{1} + w_{2} + a_{12})<0$.
\bulletend
\end{lemma}

\begin{remark} \label{rem:ellipse}
\thmtitle{Geometric and other interpretations of the conditions for the 
existence of periodic opinions}
	 The condition given in Lemma~\ref{lem:necessary_condition_for_periodic_orbits} ensures existence of the following ellipse in the phase plane,
	 	\begin{equation} \label{eq:ellipse}
	 		\frac{z_{1}^{2}}{r_{1}} + \frac{z_{2}^{2}}{r_{2}} = -\frac{(w_{1} + w_{2} + a_{12}+a_{21})}{3} \rdef \upsilon>0.
	 	\end{equation}
  From the proof of 
  Lemma~\ref{lem:necessary_condition_for_periodic_orbits}, we see that 
  the $\mathsf{div}(\mathbf{f}(\boldz))$ is equal to zero on the 
  ellipse and positive (negative) inside (outside) the ellipse. Thus, 
  any periodic orbit of \eqref{eq:two_agent_dynamics} in the phase 
  plane should necessarily intersect the above ellipse.
  
The oscillatory behavior of opinions appears only when at least one of the agents has an antagonist influence on the other. This is similar to the so-called 
  \emph{boomerang effect}~\cite{1967_RA_JM_boomerang_effect}, where at 
  least one of the agents has an antagonist influence on another agent. 
  As a result, the agent on whom there is an antagonist influence 
  shifts its opinion away from the other agent. In such situations, the 
  opinions of the agents could converge to a disagreement equilibrium 
  or could possibly keep on oscillating forever never converging to an 
  equilibrium. 
  \remend
\end{remark}

Now that we have dealt with necessary conditions for a periodic solution of~\eqref{eq:two_agent_dynamics} to exist, we conclude this section by giving sufficient conditions for the agents to exhibit periodic opinion profiles. More specifically, we give sufficient conditions on the conformity weights $a_{12}$, $a_{21}$ for a Hopf bifurcation to exist (proof in Appendix).

\begin{theorem}{(\textbf{\emph{Existence of Hopf
        bifurcation}})} \label{thm:hopf}
    For the two-agent dynamics \eqref{eq:two_agent_dynamics} with
    $a_{21} \in \real$ as the bifurcation parameter, let
    $\kappa_{i}:= \left[\wi+\frac{3m_i^2}{\ri}\right] >0, i \in
    \{1,2\}$. Let $m_{1}=m_{2}=m \in \real$ and
    $\left[a_{12}(\kappa_{2}-\kappa_{1})-\kappa_{1}^{2}\right]> 0$.
    Then, a family of periodic orbits of \eqref{eq:two_agent_dynamics}
    bifurcates from the consensus equilibrium $\boldz^{*}=m\boldone$
    at $a_{21}=a_{21}^{*} := - (\kappa_1 + \kappa_2 + a_{12})$. \remend
\end{theorem}                      

\section{Simulations} \label{sec:sims}

In this section, we present some simulations to demonstrate our 
analytical results. All the simulation results were generated using 
MATLAB and the ODE 45 solver. Through out this section, the 
$i^{\text{th}}$ element of any parameter data vector corresponds to 
agent $i$ and we use $\approx(=)$ signs for indicating approximate 
(exact) values of the provided data.

In the first set of simulations, we consider a group of 6 agents 
forming opinions according to~\eqref{eq:dynamics}. We assume that the 
agents are connected via an influence network that is shown in 
Figure~\ref{fig:network}. The matrix of influence weights $\adjmat$ is 
equal to the adjacency matrix associated with the graph shown in 
Figure~\ref{fig:network}.
\begin{figure}[htb]
  \centering
  \includegraphics[scale=0.4,trim = 0in 7.3in 2.3in 0.3in, clip]{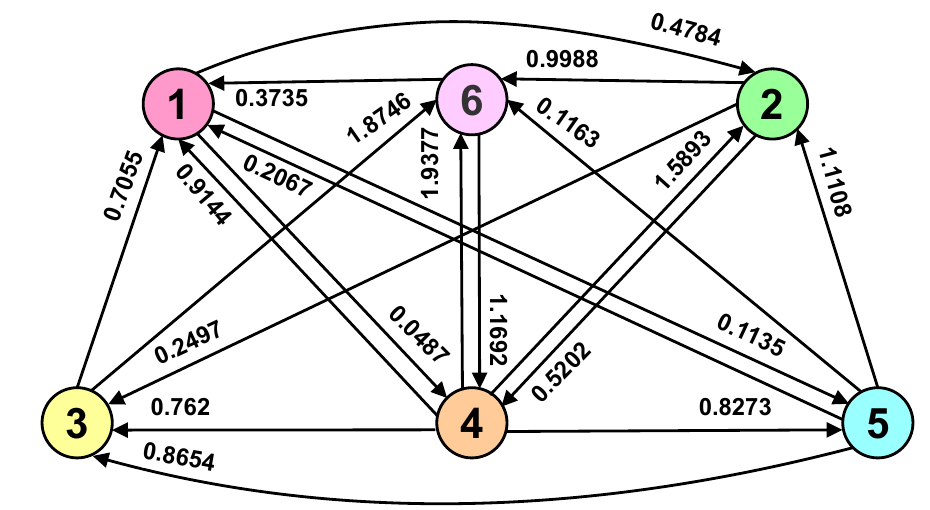}
  \caption{A social network consisting of 6 agents. The direction of 
  any link denotes the direction of influence and the number near 
  arrowhead of any directed link $(k,i)$ represents the corresponding 
  link weight $\aik$.}
  \label{fig:network}
\end{figure}
In Figure~\ref{fig:consensus}, we illustrate the case where the
opinions of all 6 agents reach a non-neutral consensus equilibrium,
with consensus value equal to 40. The model parameters used to
simulate this case are as follows. The vectors of initial opinions
$\zi(0)$, importance weights $w_i$, resources $\ri$, agents' internal
preferences $\di$ of all six agents are
\begin{align*}
  \mathbf{z}_{0}&\approx 
  \begin{bmatrix}
    43.90
    &36.34
    &49.00
    &30.69
    &38.77
    &37.63
  \end{bmatrix},
\\
\mathbf{w} &\approx 
\begin{bmatrix}
  1.09
  &2.98
  &2.59
  &1.82
  &1.01
  &2.65
\end{bmatrix},
\\
\mathbf{r}&\approx
\begin{bmatrix}
  317.43
  &814.54
  &789.07
  &852.26
  &505.64
  &635.66
\end{bmatrix},
\\
\mathbf{p} &\approx  %
\begin{bmatrix}%
  225.22
  &66.40
  &71.33
  &81.16
  &165.39
  &77.99
\end{bmatrix},
\end{align*}
respectively. For these parameters, $m_{i}=40\:,\:\forall i \in \agt$, 
which is equal to the consensus value. Thus, this simulation verifies 
Theorem~\ref{thm:consensus_iff}. For each agent $i \in \agt$, let
$z_{i}^{\infty}$ denote its asymptotic opinion value. The 
vector whose each element is the absolute difference between 
an agent's final consensus opinion and its preference opinion 
($|\zi^{\infty}-\di|$) and the consensus dominance weights $\sigma_i$ are
\begin{align*}
  \mathbf{d}&\approx
   \begin{bmatrix}                            
    185.22
    &26.40
    &31.33
    &41.16
    &125.39
    &37.99
  \end{bmatrix}
\\
\boldsymbol{\sigma}&\approx
\begin{bmatrix}
  345.5
  &2424.4
  &2042.5
  &1555.1
  &510.4
  &1684.8
\end{bmatrix}.
\end{align*}

From this data, we can verify that the dominance claim in
Proposition~\ref{prop:dominance} is satisfied in this case.

\begin{figure}[htb]
  \begin{subfigure}{0.24\textwidth}
    \includegraphics[trim = 1.3in 3.2in 1.4in 3.2in, clip, 
    width=\textwidth]{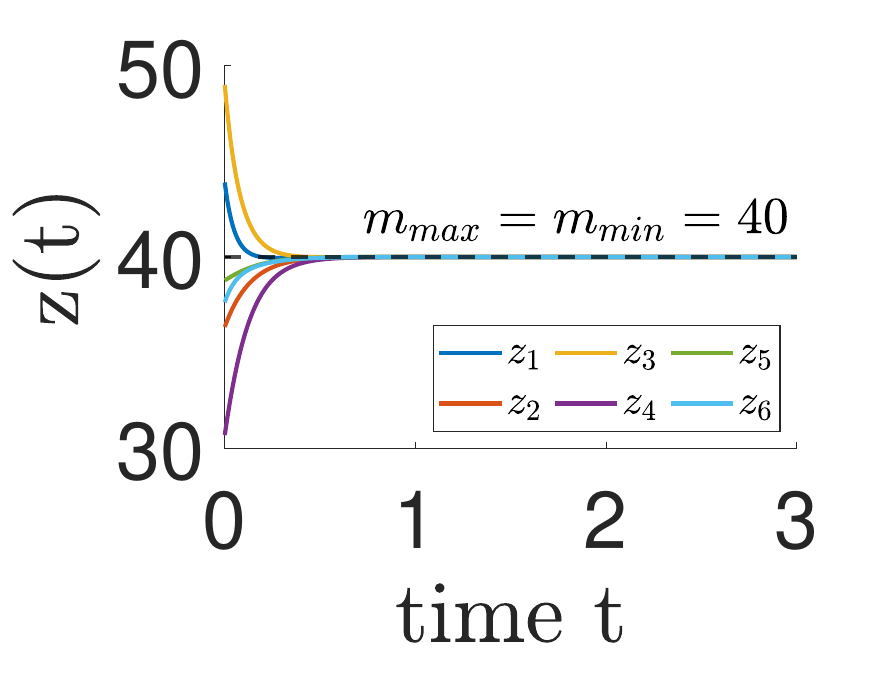}
    \caption{}
    \label{fig:consensus}
  \end{subfigure}
  \begin{subfigure}{0.24\textwidth}
    \includegraphics[trim = 1.3in 3.2in 1.4in 3.2in, clip, 
    width=\textwidth]{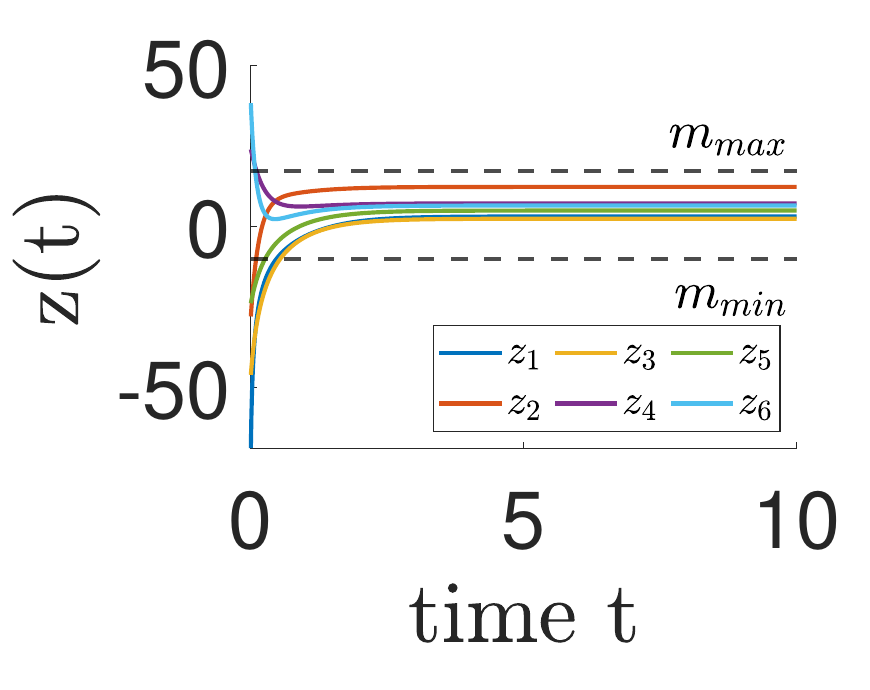}
    \caption{}
    \label{fig:convg_to_non_cnsus_eqm}
  \end{subfigure}
  \caption{Convergence of opinions. (a) Consensus equilibrium. (b) 
  Disagreement equilibrium within $\Mc{M}^{n}$.}
\end{figure}

Figure~\ref{fig:convg_to_non_cnsus_eqm} depicts the scenario where the 
opinions of the 6 agents with no antagonistic relations converge to a 
disagreement equilibrium in the compact set $\Mc{M}^{n}$. The model 
parameters in this case are,
\begin{align*}
  \mathbf{z}_{0}&\approx 
  \begin{bmatrix}
    -69.02
    &-28.03
    &-46.22
    &24.01
    &-23.92
    &38.38
  \end{bmatrix}
\\
\mathbf{w}&\approx
\begin{bmatrix}
  0.22
 &3.48
 &0.56
 &2.56
 &0.95
 &2.62
\end{bmatrix}
\\
\mathbf{r}&\approx
\begin{bmatrix}
  317.43
  &814.54
  &789.07
  &852.26
  &505.64
  &635.66
\end{bmatrix}
\\
\mathbf{p}&\approx
\begin{bmatrix}
  -18.31
  &18.92
  &-12.43
  &6.68
  &3.46
  &7.00
\end{bmatrix}.
\end{align*}
The agents do not achieve consensus because
\begin{equation*}
  \mathbf{m}\approx
  \begin{bmatrix}
    -8.78
    &17.14
    &-10.10
    &6.56
    &3.38
    &6.81
  \end{bmatrix},
\end{equation*}
which violates the necessary condition for consensus given in
Theorem~\ref{thm:consensus_iff}. From 
Figure~\ref{fig:convg_to_non_cnsus_eqm}, it can be seen that every 
agent's opinion converges to a value in the set $\mathcal{M}\approx  
[-10.10,17.14]$ which verifies the results of 
Theorem~\ref{thm:exis_uniq_of_eqm} and 
Proposition~\ref{prop:ub_c_pos}.   

\par Figure~\ref{fig:limit_cycle} demonstrates the oscillatory behavior 
exhibited by opinions of two agents 
under~\eqref{eq:two_agent_dynamics}. The model parameters for this case 
are, $\mathbf{z_{0}}\approx \begin{bmatrix}
	-0.0349   &-0.0039
\end{bmatrix}$, $[a_{12} \:\: a_{21}] \approx \begin{bmatrix}
	-6.6667 & 3.3267
\end{bmatrix}$, $\mathbf{w} \approx \begin{bmatrix}
	2 & 1.3333
\end{bmatrix}$, $\mathbf{p}=\begin{bmatrix}
	0 & 0
\end{bmatrix}$ and $\mathbf{r}=\begin{bmatrix}
	10 & 5
\end{bmatrix}$. Consider \eqref{eq:two_agent_dynamics} with the above 
parameter values and let $a_{21}$ be the bifurcation parameter. Notice 
that since $p_{1}=p_{2}=0$, neutral consensus $\boldz^{*}=\boldzero 
\in\eqpt\:,\:\forall a_{21} \in \real$. These values satisfy the 
assumptions of Theorem~\ref{thm:hopf} ensuring the existence a Hopf 
bifurcation at $a_{21}^{*} = 3.3333$. The eigenvalues of the Jacobian 
$\mathbf{J}(\boldz^{*},a_{21}^{*})$ are $\pm 0.6667j$. As a result, a 
family of periodic orbits bifurcates out of the neutral consensus 
equilibrium at $a_{21}^{*}$ and the periodic solution shown in 
Figure~\ref{fig:limit_cycle} is a member of this family corresponding 
to $a_{21}= 3.3267 < a_{21}^{*}$. Figure~\ref{fig:limit_cycle_ellipse_cut} 
shows the corresponding trajectory in the phase plane intersecting the 
ellipse defined in \eqref{eq:ellipse}. These simulations support the 
claim of 
Lemma~\ref{lem:necessary_condition_for_periodic_orbits} and its 
interpretation made in Remark~\ref{rem:ellipse}.
\begin{figure}[htb]
  \begin{subfigure}{0.24\textwidth}
    \centering
    \includegraphics[trim = 1.3in 3.2in 1.4in 3.2in, 
    clip,width=\textwidth]{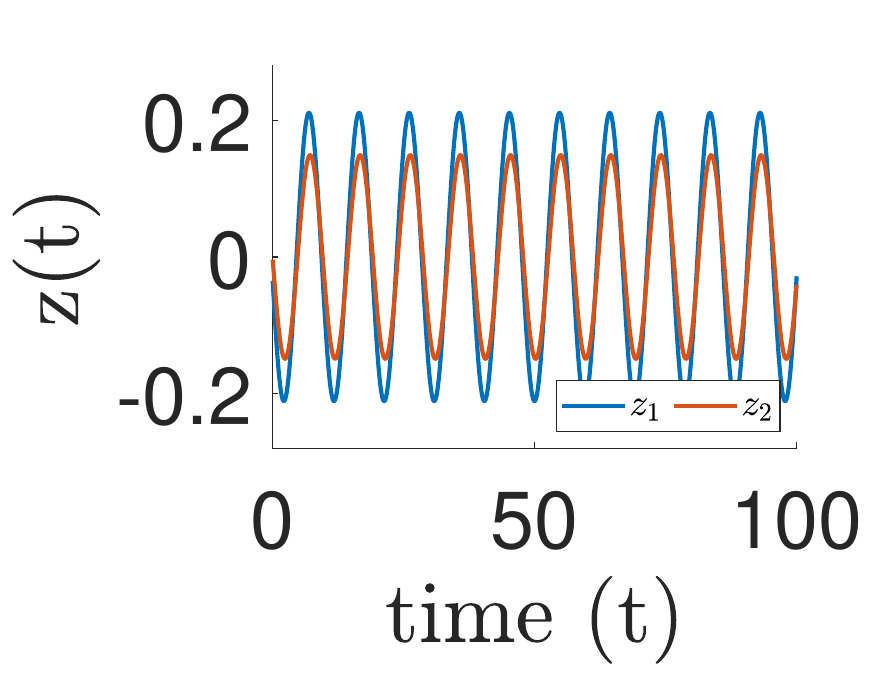}
    \caption{}
    \label{fig:limit_cycle}
  \end{subfigure}
  \begin{subfigure}{0.24\textwidth}
    \centering
    \includegraphics[trim = 1.3in 3.2in 1.4in 3.4in, 
    clip,width=\textwidth]{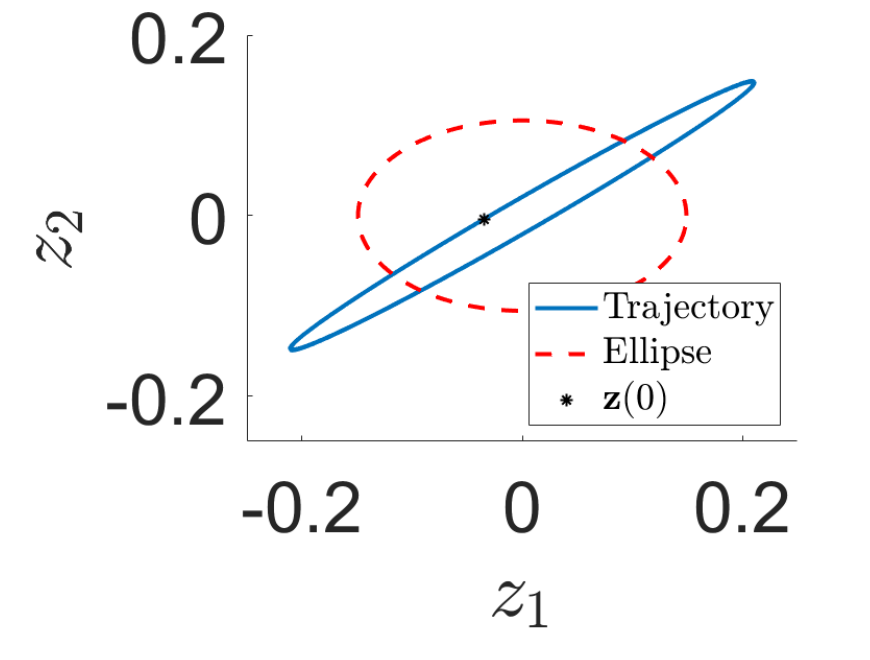}
    \caption{}
    \label{fig:limit_cycle_ellipse_cut}
  \end{subfigure}
  \caption{Oscillatory behavior of opinions. (a) Opinion trajectory. 
  (b) Intersection of corresponding trajectory with the ellipse.}
\end{figure}

\section{Conclusions} \label{sec:conclusions}
We proposed a non-linear model of opinion dynamics to capture the
effect of heterogeneous resources available to the agents on their
opinions.  In contrast to our prior work~\cite{2023_PW_NM_PT}, we
dealt with general social networks with (possibly) antagonistic
relations.  We showed ultimate boundedness of opinions and provided
sufficient conditions for the dynamics to have a globally
exponentially stable equilibrium point.  We also provided necessary
and sufficient condition for the existence of a consensus equilibrium
and quantified social dominance at consensus.  Further, we showed that
the set of Nash equilibria of the opinion formation game is a subset
of the set of equilibrium points of the dynamics and provided
conditions for these two sets to coincide.  In the absence of
antagonistic relations, we gave stronger results. We quantified the
quality of the Nash equilibria with respect to two commonly used
prices of anarchy ($\poa$), provided bounds on these $\poa$'s in terms
of the satisfaction ratios and proved that converging to a consensus
equilibrium is a socially optimal outcome.  Finally, we analyzed the
periodic behavior of opinions exhibited by the proposed dynamics for
the case of two agents. We provided necessary conditions for periodic
solutions to exist and sufficient conditions for a Hopf bifurcation to
occur at the consensus equilibrium. Future research directions include
extensions of the model to a multi-topic scenario, analysis of
periodic behavior exhibited by opinions in the presence of
antagonistic relations for a general $n$-agent case, and exploration
of a more general class of utility functions and resource penalty
functions.

\appendix

\subsection{Proof of Theorem \ref{thm:exis_uniq_of_eqm}}

To prove this, we show that the dynamics~\eqref{eq:dynamics} is 
strongly contracting. Let $\mathbf{J}(\boldz)\ldef
\left[\frac{\partial\mathbf{f}(\boldz)}{\partial\boldz}\right]$ 
denote the Jacobian matrix of~\eqref{eq:dynamics}. The $(i,j)\tth$ 
element of $\mathbf{J}(\boldz)$ is
	\begin{align*}
		[\mathbf{J}(\boldz)]_{ij}=
		\begin{cases}
			\displaystyle
			-\wi - \left( \agtsumknoti \aik
			\right)-\frac{3\zi^{2}}{\ri}\,, 
			& \mathrm{if}\, i= j; \\ 
			\displaystyle \aij, &  \mathrm{if}\, i \neq j .
		\end{cases}
	\end{align*}
The induced $\infty-$log norm of $\mathbf{J}(\boldz)$ is:
	\begin{equation*}
		\mu_{\infty}(\mathbf{J}(\boldz))=\max_{i \in \agt} \bigg[-\wi+\sum_{k \in \Nie}2|\aik| -\frac{3\zi^{2}}{\ri}\bigg].
	\end{equation*}
Under Assumption~\ref{asmp:weak_antagonist_relations}, it can be seen 
that $\exists \,\: k>0$ such that 
$\mu_{\infty}(\mathbf{J}(\boldz))<-k\:;\:\forall \boldz \in \realn$. 
Now, Theorem~\ref{thm:ultimate_bound} guarantees existence 
of a set $\Omega \subset\realn$ which is convex, closed and 
positively invariant under the dynamics \eqref{eq:dynamics}. Further, 
from the discussion below Theorem~\ref{thm:ultimate_bound}, we know that no 
equilibrium exists outside $\Omega$. Thus, 
from Theorem~\ref{thm:contraction}, the dynamics~\eqref{eq:dynamics} is 
strongly contracting in $\Omega$ and the unique 
equilibrium is globally exponentially stable. This completes the proof. $\hfill \blacksquare$

\subsection{Proof of Proposition \ref{prop:ub_c_pos}} 

Notice that under Assumption~\ref{asmp:no_antagonist}, the 
dynamics~\eqref{eq:dynamics} reduces 
to~\eqref{eq:dynamics_no_antagonists}. First, we prove positive 
invariance of $\Mc{M}^{n}$ under~\eqref{eq:dynamics_no_antagonists} by 
inspecting the vector field at $\boldz$ on the boundary of 
$\Mc{M}^{n}$. Consider any $i \in \agt$ such that $z_i = m_{\max}$. 
From~\eqref{eq:SC_sign}, we see 
that $\Sf_i(z_i) \leq 0$ and $\Cf_i^{+}(\boldz) \leq 0$, as $m_i \leq 
m_{\max}$ and $\bar{z}_i$ is a convex combination of $\{z_j\}_{j=1}^n$ 
and $\boldz \in \Mc{M}^{n}$. Thus, $f_i(\boldz) \leq 0$. Similarly, for 
any $i \in \agt$ such that $z_i = m_{\min}$, $f_i(\boldz) \geq 0$. 
This implies that $\Mc{M}^n$ is positively invariant 
under~\eqref{eq:dynamics_no_antagonists}.

Now, we show that $\boldz(t)$ converges to $\Mc{M}^n$. Note that
Assumption~\ref{asmp:no_antagonist} is a special case of
Assumption~\ref{asmp:weak_antagonist_relations} and hence
Theorem~\ref{thm:exis_uniq_of_eqm} guarantees that there is a unique
globally exponentially stable equilibrium point $\boldz^* \in
\eqpt$. We know that under Assumption~\ref{asmp:no_antagonist},
$\zbari$ is a convex combination of $\{z_k\}_{k \in \Ni}$ and
that~\eqref{eq:SC_sign} holds.  So, $\forall \boldz \notin \Mc{M}^n$,
$\exists \, i \in \agt$ such that $\dot{z}_i \neq 0$. Thus, the unique
globally exponentially stable equilibrium point
$\boldz^* \in \Mc{M}^n$ and hence $\boldz(t)$ converges to $\Mc{M}^n$.

Next, we prove the last claim in the result, i.e., for the special case 
of $m_{\min} < m_{\max}$. We have already seen that under 
Assumption~\ref{asmp:no_antagonist}, the unique globally exponentially 
stable equilibrium point $\boldz^* \in \Mc{M}^n$. The following 
statement is a direct consequence of~\eqref{eq:SC_sign}.
\begin{enumerate}[label=\textbf{(S\arabic*)}] 
  \item For $i \in \agt$, $\zi^{*}=m_{\max}$ if and only if $m_{i} 
  = m_{\max}$ and $z_j^{*} = m_{\max}, \ \forall j \in \Ni$.
\end{enumerate}
From the necessary condition on $m_i$ in \textbf{(S1)}, we can say that
\begin{enumerate}[resume,label=\textbf{(S\arabic*)}] 
  \item $\zi^* < m_{\max}, \ \forall i \in (\agt \setminus \agtmmax)$.
\end{enumerate}
A further consequence of \textbf{(S1)} is that for $i \in \agtmmax$,
if $\exists j \in \Ni \cap ( \agt \setminus \agtmmax )$ then
$\zbari^* < m_{\max}$ and hence $z_i^* < m_{\max}$.  This fact along
with \textbf{(S1)} and \textbf{(S2)} imply that for $i \in \agtmmax$,
$\zi^* < m_{\max}$ if and only if there is a directed walk in $\grph$
from $j \in (\agt \setminus \agtmmax)$ to $i$. We can make similar
observations corresponding to $m_{\min}$. From all these observations,
we can finally say that
\begin{enumerate}[resume,label=\textbf{(S\arabic*)}] 
  \item $\boldz^*$ lies in the interior of $\Mc{M}^n$ if and only if  
  $\forall i \in \agtmmax$, $\exists$ a directed walk in $\mathsf{G}$ 
  starting from $j \in \agt \setminus \agtmmax$ to $i$ and $\forall 
  i \in \agtmmin$, $\exists$ a directed walk in $\mathsf{G}$ starting 
  from $j \in \agt \setminus \agtmmin$ to $i$.  
\end{enumerate}
Convergence to $\Mc{M}^n$ occurs in finite time if and only if the 
globally exponentially stable equilibrium $\boldz^*$ lies in the 
interior of $\Mc{M}^n$. This completes the proof of the result. $\hfill 
\blacksquare$

\subsection{Proof of Lemma \ref{lem:convexcost}}
\par Suppose Assumption~\ref{asmp:no_antagonist} holds. Then $\forall i 
\in \agt$, it can be verified using Gerschgorin disc theorem that the 
Hessian of $\costi(\boldz)$ is positive semidefinite, $\forall \boldz 
\in \realn$. Hence, $\costi(\boldz)$ is convex. Now, 
$m_{i}\boldone$ satisfies the first order necessary condition for a 
point to be a local minimizer of $\costi(\boldz)$. Since 
$\costi(\boldz)$ is convex, $m_{i}\boldone$ is also a minimizer of 
$\costi(\boldz)$. $\hfill \blacksquare$

\subsection{Proof of Theorem \ref{thm:poa_bound}}
We get the lower bound on $\poae$ and $\poau$ from the Definition~\ref{def:poa}. To get an upper bound on $\poae$, note that 
\begin{equation*}
	\min_{\boldz \in \realn }\max_{i \in \agt} \costi(\boldz) \geq \max_{i \in \agt}\min_{\boldz \in \realn }\costi(\boldz).
\end{equation*}
Using this inequality and Lemma~\ref{lem:convexcost}, we can upper bound $\poae$ defined in Definition~\ref{def:poa} as, $\forall j \in \agt$,
\begin{align*}
	\poae &\leq \frac{\max_{i \in \agt} 
	\costi(\boldz^{*})}{ \max_{i \in 
	\agt}\costi(m_{i}\boldone)} \leq \max_{i \in \agt} \bigg[\frac{ 
	\costi(\boldz^{*})}{\costj(m_{j}\boldone)} \bigg]\, ,
\end{align*}
for any $j \in \agt$. Choosing $j=i$, 
Definition~\ref{def:satisfactionratio} and Lemma~\ref{lem:convexcost} 
implies the bound on $\poae$ given in~\eqref{ineq:poa_bound}.
Now, consider $\poau$ as in Definition~\ref{def:poa}. From Lemma~\ref{lem:convexcost}, we have,
\begin{equation*}
	\min_{\boldz \in \realn} \sum_{i \in \agt} \costi(\boldz) \geq \sum_{i \in \agt}\costi(m_{i}\boldone).
\end{equation*}
Using this we get, $\forall j \in \agt$,
\begin{align*}
	\poau &\leq \frac{\sum_{i \in \agt}\costi(\boldz^{*})}{\sum_{i \in \agt}\costi(m_{i}\boldone)} = \sum_{i \in \agt} \Bigg[\frac{\costi(\boldz^{*})}{\sum\limits_{j \in \agt}\costj(m_{j}\boldone)}\Bigg] \leq \sum_{i \in \agt} 
	\frac{\costi(\boldz^{*})}{\costj(m_{j}\boldone)}.   
\end{align*}
Choosing $j=i$, Definition~\ref{def:satisfactionratio} and 
Lemma~\ref{lem:convexcost} implies the bound on $\poau$ given 
in~\eqref{ineq:poa_bound}.
$\hfill \blacksquare$

\subsection{Proof of Lemma \ref{lem:equal_periods} }
To prove this, rewrite \eqref{eq:two_agent_dynamics} as, 
  \begin{equation*}
    \dot{z}_{1}- \Sf_{1}(z_{1})+a_{12} z_{1} = a_{12} z_{2}, \
    \dot{z}_{2}- \Sf_{2}(z_{2})+a_{21} z_{2} = a_{21}z_{1}
  \end{equation*}
Now, suppose $z_{1}(t)$ and $z_{2}(t)$ are periodic with fundamental
periods $T_{1}>0$ and $T_{2}>0$ respectively. Then, $ \Sf_1(z_1(t))$
and $\dot{z}_1$ (resp. $ \Sf_2(z_2(t))$ and $\dot{z}_2$) are also
periodic with fundamental period $T_1$ (resp. $T_2$). Then, we can see
that $\exists \,\: m,n \in \nat$ such that $T_{1} = m \, T_{2}$ and
$T_{2} = n \, T_{1}$. Thus, $m=n=1$. $\hfill \blacksquare$

\subsection{Proof of Lemma 
\ref{lem:necessary_condition_for_periodic_orbits}}
From the discussion just above Lemma~\ref{lem:necessary_condition_for_periodic_orbits}, it is necessary that at least one of $a_{12}$ or $a_{21} < 0$ for periodic solutions to exist. Further, it is also necessary that both influence weights are non-zero. Otherwise, from Lemma \ref{lem:eqm_isolated_agents}, the socially closed agent's 
opinion would converge to its own $m_{i}$ value. The opinion of other agent then would always converge 
to an equilibrium.
Now, it remains to show that $(a_{21}+ w_{1} + w_{2} + a_{12})<0$
is necessary for existence of periodic solutions. We prove the inverse using Bendixson's criteria.
Note the divergence of the vector field \eqref{eq:two_agent_dynamics} at $\boldz \in \real^{2}$ is,
\begin{equation*}
	\mathsf{div}(\mathbf{f}(\boldz)) = -(w_{1} + w_{2} 
	+ a_{12}+a_{21}) - \frac{3z_{1}^{2}}{r_{1}} - 
	\frac{3z_{2}^{2}}{r_{2}}.
\end{equation*}

Suppose $(a_{21} + w_{1} + w_{2} + a_{12})>0$.
Then, $\mathsf{div}(\mathbf{f}(\boldz))< 0$, $\forall \boldz \in \real^{2}$. Thus, by Bendixson's criteria \cite[Lemma 2.2]{2002_Khalil_non_linear_book} there are no periodic orbits lying entirely in $\real^{2}$.\\ 
Now, consider the case when 
	$(a_{21} + w_{1} + w_{2} + 
	a_{12})=0$.
In this case, $\mathsf{div}(\mathbf{f}(\boldz))< 
0$, $\forall \boldz \in \real^{2}\setminus \{\boldzero\}$ and 
$\mathsf{div}(\mathbf{f}(\boldz))= 0$ iff $\boldz=\boldzero$. Suppose 
there exists a periodic orbit, then in the simply connected region $\mathcal{S} 
\subset \real^{2}$ enclosed by the periodic orbit, $\iint_{\mathcal{S}} 
\left(\mathsf{div}(\mathbf{f}(\boldz))\right)\mathsf{d}z_{1}\mathsf{d}z_{2}
= 0$. However, this cannot happen unless $\mathcal{S} = \{\boldzero\}$, which 
cannot correspond to a periodic solution. 
Hence, again, periodic orbits cannot exist.
$\hfill \blacksquare$ 

\subsection{Proof of Theorem \ref{thm:hopf}} 	Under the assumption 
$m_{1}=m_{2}=m$, it follows from Theorem~\ref{thm:consensus_iff} that 
$\boldz^{*}=m\boldone \in \eqpt\:;\:\forall a_{21}\in \real$.
The Jacobian of \eqref{eq:two_agent_dynamics} evaluated about $\boldz^{*}$ is,
	\begin{equation*}
		\mathbf{J}(\boldz^{*},a_{21}) = \begin{bmatrix}
			(-\kappa_{1}-a_{12}) & a_{12} \\
			a_{21} & (-\kappa_{2}-a_{21})
		\end{bmatrix}.		
	\end{equation*}
Under the stated assumptions, if $a_{21}=a_{21}^{*}$ then the
eigenvalues of $\mathbf{J}(\boldz^{*},a_{21}^{*})$ are purely
imaginary. This satisfies the first assumption of the Hopf bifurcation
theorem \cite[Theorem~ 3.4.2]{2013_Guckenheimer_book}. The eigenvalues
$\lambda(a_{21})$ of $\mathbf{J}(\boldz^{*},a_{21})$ which are purely
imaginary at $a_{21}=a_{21}^{*}$ vary smoothly with the $a_{21}$. For
values of $a_{21}$ sufficently close to $a_{21}^{*}$, the real part of
complex conjugate eigenvalue pair can be given as
	$Re(\lambda(a_{21})) = -0.5\left[\kappa_{1}+\kappa_{2}+a_{12}+a_{21}\right] $.
 Then, the derivative of $Re(\lambda(a_{21}))$ with respect to $a_{21}$ evaluated at $a_{21}^{*}$ is
	$\frac{\mathsf{d}}{\mathsf{d} a_{21}}\left[Re(\lambda(a_{21}))\right]\mid_{a_{21}=a_{21}^{*}} = -0.5 \neq 0$.
Hence, the second assumption of the Hopf theorem is also satisfied. The claim now follows directly from the Hopf bifurcation theorem. $\hfill \blacksquare$


	\bibliographystyle{IEEEtran}
	\bibliography{references}	
	
	\vspace{-10mm}
\begin{IEEEbiography}[{\includegraphics[width=1in,
		height=1.25in,clip,keepaspectratio] {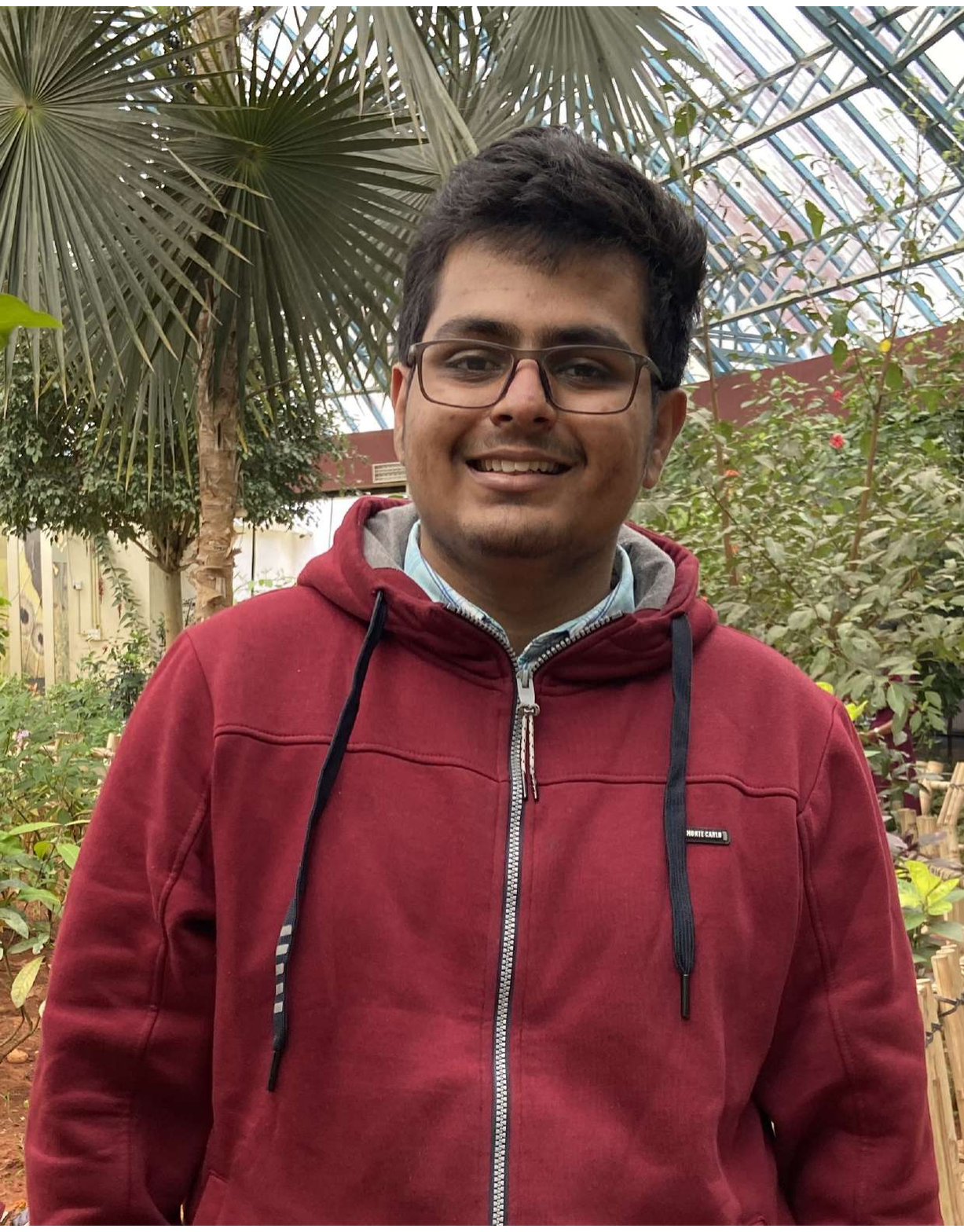}}]{Prashil Wankhede}(Student Member IEEE) received the B.Tech. degree in Electrical Engineering from Sardar Patel College of Engineering, Mumbai, India, in 2021. He is currently a Ph.D. student in the Electrical Engineering Department at the Indian Institute of Science, Bangalore, India. His research interests include multi-agent systems, opinion dynamics, game theory and non-linear control.
\end{IEEEbiography}

\vskip -2\baselineskip plus -1fil

\begin{IEEEbiography}[{\includegraphics[width=1in,
    height=1.25in,clip,keepaspectratio] {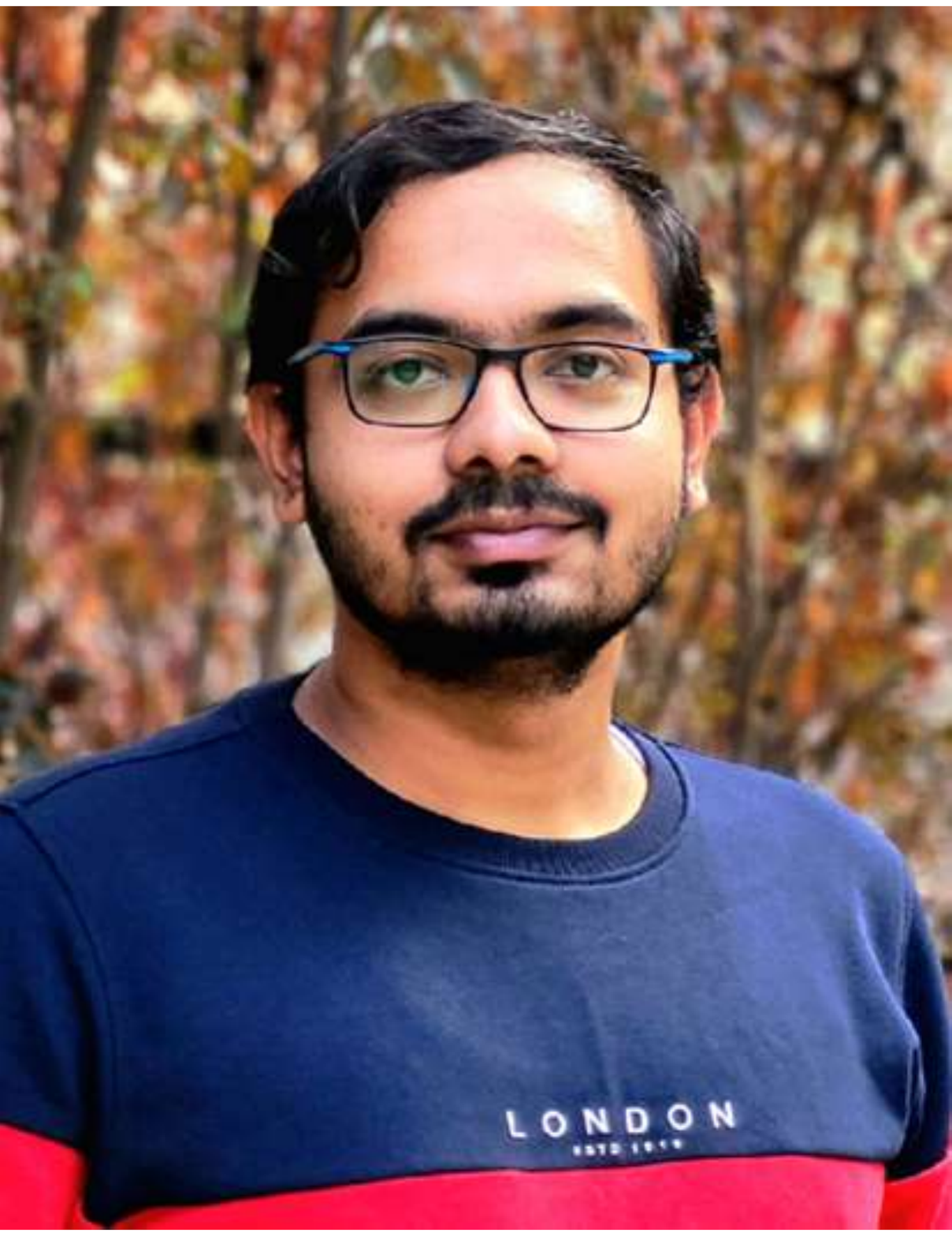}}]{Nirabhra Mandal} (Student Member IEEE) received the B.Tech. degree in Electrical Engineering from the Institute of Engineering and Management, Kolkata, India, in 2017 and the M.Tech.(Res) degree in Electrical Engineering from the Indian Institute of Science, Bengaluru, India, in 2021.
He is currently a Ph.D. student in the Mechanical and Aerospace Engineering Department at University of California, San Diego, USA. His research interests include multiagent systems, distributional robust optimzation, population games, evolutionary dynamics on networks, and nonlinear control.
\end{IEEEbiography}

\vskip -2\baselineskip plus -1fil

\begin{IEEEbiography}[{\includegraphics[width=1in,
    height=1.25in,clip,keepaspectratio] {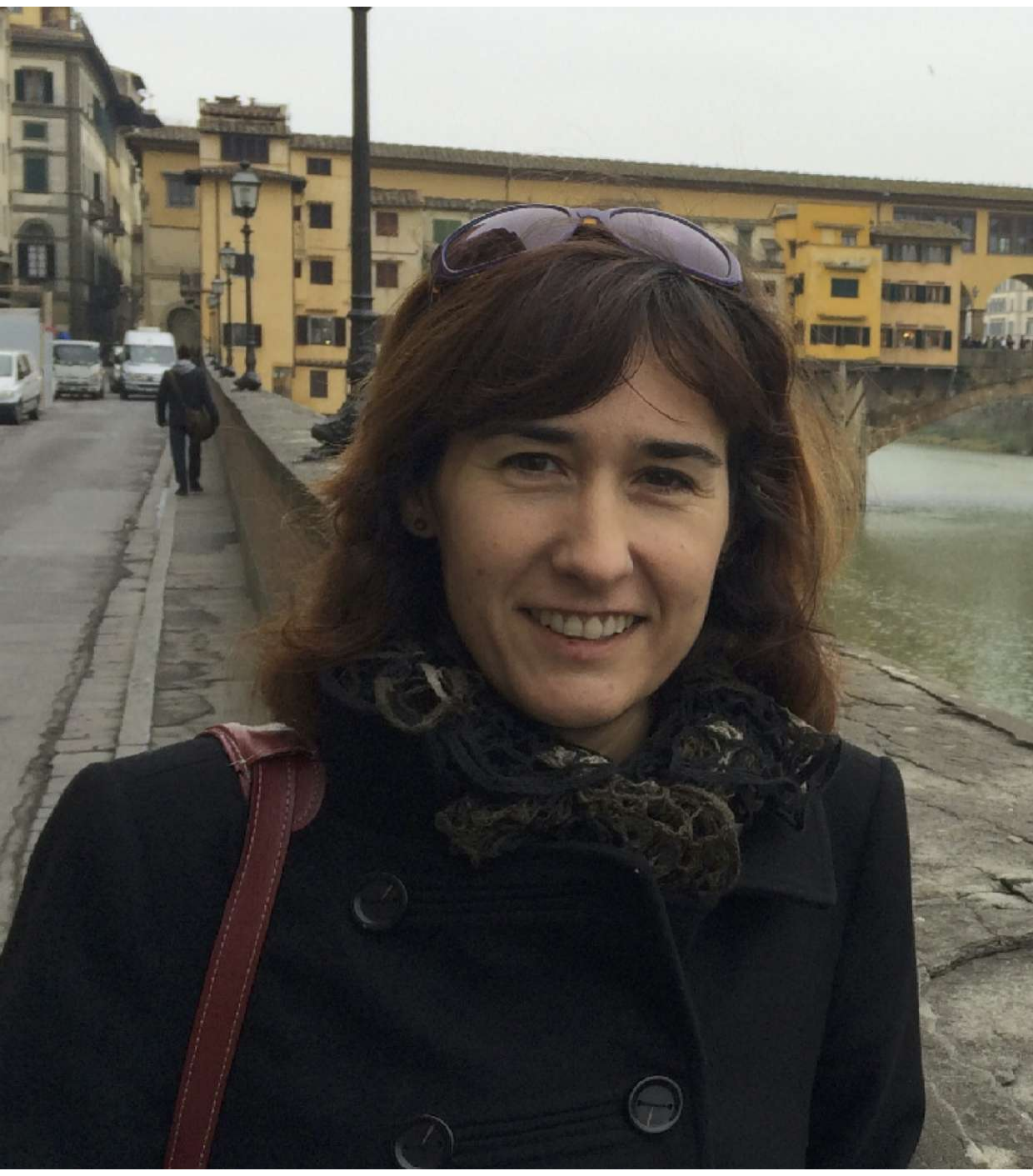}}]{Sonia Mart\'{i}nez} (Fellow IEEE) is a Professor of Mechanical and Aerospace Engineering
at the University of California, San Diego, CA, USA. She received her
Ph.D. degree in Engineering Mathematics from the Universidad Carlos
III de Madrid, Spain, in May 2002. She was a Visiting Assistant
Professor of Applied Mathematics at the Technical University of
Catalonia, Spain (2002- 2003), a Postdoctoral Fulbright Fellow at the
Coordinated Science Laboratory of the University of Illinois,
Urbana-Champaign (2003-2004) and the Center for Control, Dynamical
systems and Computation of the University of California, Santa Barbara
(2004-2005). Her research interests include the control of networked
systems, multi-agent systems, nonlinear control theory, and planning
algorithms in robotics. She became a Fellow of IEEE in 2018. She is a
co-author (together with F. Bullo and J. Cortés) of ‘‘Distributed
Control of Robotic Networks’’ (Princeton University Press, 2009). She
is a co-author (together with M. Zhu) of ‘‘Distributed
Optimization-based Control of Multi-agent Networks in Complex
Environments’’ (Springer, 2015). She is the Editor in Chief of the
recently launched CSS IEEE Open Journal of Control Systems.
\end{IEEEbiography}

\vskip -2\baselineskip plus -1fil

\begin{IEEEbiography}[{\includegraphics[width=1in,
    height=1.25in,clip,keepaspectratio] {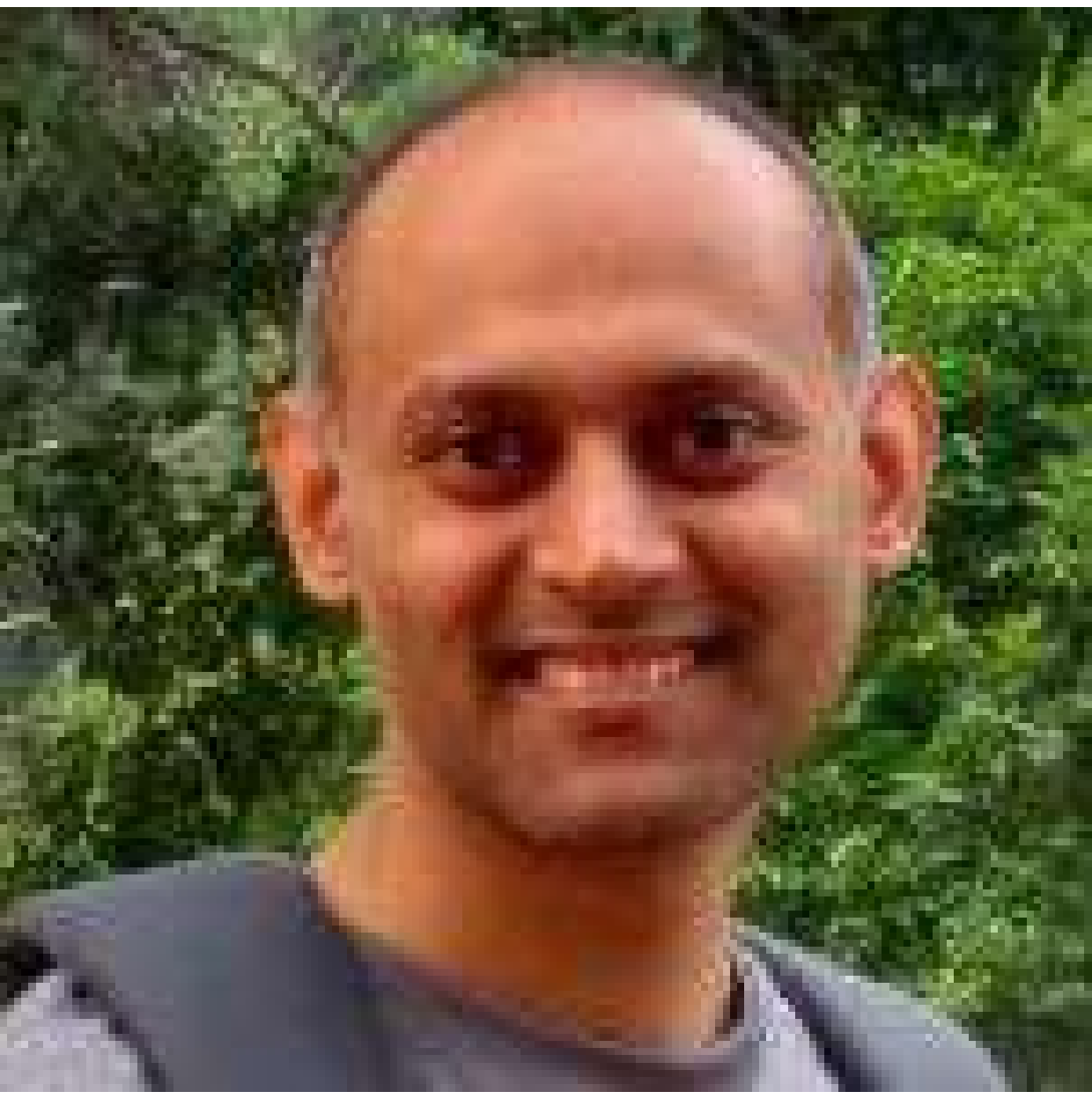}}]{Pavankumar
    Tallapragada} (S'12-M'14) received his Ph.D. degree in Mechanical
  Engineering from the University of Maryland, College Park in
  2013. He was a Postdoctoral Scholar in the Department of Mechanical
  and Aerospace Engineering at the University of California, San Diego
  from 2014 to 2017. He is currently an Associate Professor in the
  Department of Electrical Engineering at the Indian Institute of
  Science, Bangalore, India. His research interests include networked
  control systems, distributed control, multi-agent systems and
  dynamics of social systems.
\end{IEEEbiography}
\end{document}